\documentclass[11p,reqno]{amsart}

\usepackage[T1]{fontenc}
\usepackage{graphicx}
\usepackage{amssymb,amsthm,amsmath,mathrsfs,bm,braket,marginnote}
\usepackage{enumerate}
\usepackage{appendix}
\usepackage[colorlinks=true, pdfstartview=FitV, linkcolor=blue, citecolor=blue, urlcolor=blue]{hyperref}
\usepackage{multirow}

\usepackage{pgf}
\usepackage{pgfplots}
\usepackage{tikz}
\usetikzlibrary{arrows,calc}
\usepackage{verbatim}
\usetikzlibrary{decorations.pathreplacing,decorations.pathmorphing}
\usepackage[numbers,sort&compress]{natbib}
\usepackage{dsfont}

\numberwithin{equation}{section}
\newtheorem{theorem}{Theorem}[section]
\newtheorem{lemma}[theorem]{Lemma}
\newtheorem{define}[theorem]{Definition}
\newtheorem{remark}[theorem]{Remark}
\newtheorem{corollary}[theorem]{Corollary}

\newtheorem{proposition}[theorem]{Proposition}

 \reversemarginpar
 
 \newcommand{\Pf}{\text{Pf}}
 \newcommand{\B}{\text{Bures}}
 \newcommand{\C}{\text{Cauchy}}
 \newcommand{\Zzzzz}{Z_{0|0;0|0}^}
 \newcommand{\Zozoz}{Z_{1|0;1|0}^}
 \newcommand{\Zzozo}{Z_{0|1;0|1}^}
 \newcommand{\Zoozz}{Z_{1|1;0|0}^}
  \newcommand{\Zozzz}{Z_{1|0;0|0}^}
  \newcommand{\Zzozz}{Z_{1|0;0|0}^}
 \newcommand{\Zzzoz}{Z_{0|0;1|0}^}
 \newcommand{\Zzzzo}{Z_{0|0;0|1}^}
  \newcommand{\Zzzoo}{Z_{0|0;1|1}^}
 \newcommand{\tN}{\tilde{N}}

 \reversemarginpar

\begin{document}

\title[Fox H-kernel and $\theta$-deformation of Cauchy two-matrix model and Bures ensemble]{Fox H-kernel and $\theta$-deformation of the Cauchy two-matrix model and Bures ensemble} 

\author{Peter J. Forrester}
\address{School of Mathematical and Statistics, ARC Centre of Excellence for Mathematical and Statistical Frontiers, The University of Melbourne, Victoria 3010, Australia}
\email{pjforr@unimelb.edu.au}

\author{Shi-Hao Li}
\address{ School of Mathematical and Statistics, ARC Centre of Excellence for Mathematical and Statistical Frontiers, The University of Melbourne, Victoria 3010, Australia}
\email{shihao.li@unimelb.edu.au}

\date{}

\dedicatory{}


\begin{abstract}
A $\theta$-deformation of the Laguerre weighted Cauchy two-matrix model, and the Bures ensemble, is introduced. Such a deformation is familiar  from the Muttalib-Borodin ensemble. The $\theta$-deformed Cauchy-Laguerre two-matrix model is a two-component determinantal point process. It is shown that the correlation kernel, and its hard edge scaled limit, can be written as the Fox H-functions, generalising the Meijer G-function class known from the study of the case $\theta= 1$. In the $\theta=1$ case, it is shown Laguerre-Bures ensemble is related to the Laguerre-Cauchy two-matrix model, notwithstanding the Bures ensemble corresponds to a Pfaffian point process. This carries over to the $\theta$-deformed case, allowing explicit expressions involving Fox H-functions for the correlation kernel, and its hard edge scaling limit, to be obtained. 
\end{abstract}
\maketitle

\section{Introduction}
The Cauchy two-matrix model has attracted attention due to its applications in multiple aspects of mathematical physics. It is not only significant in studies of classical integrable systems theory like Degasperis-Procesi peakons \cite{lundmark2005} and Toda lattice of CKP type \cite{li2018}, but important too in random matrix theory \cite{bertola2009,bertola2014}. This model is defined on some subset of the configuration space $\mathbb{R}^N\times\mathbb{R}^N$, as specified by an
eigenvalue probability density function (PDF) of the form
\begin{equation}\label{1}
\frac{\prod_{1\leq j<k\leq N}(x_k-x_j)(y_k-y_j)}{\prod_{j,k=1}^N(x_j+y_k)}\prod_{1\leq j<k\leq N}(x_k-x_j)(y_k-y_j)\prod_{i,j=1}^N\omega_1(x_i)\omega_2(y_j),
\end{equation}
where $\omega_1$ and $\omega_2$ are two proper (i.e.~non-negative) weight functions. An overall normalisation in (\ref{1}) is implicit and not shown
explicitly; this convention will be used throughout this section. The correlations $\rho_{\ell_1,\ell_2}(x_1,\dots,x_{\ell_1};
y_1,\dots,y_{\ell_2})$ for the subset of variables is given by an $(\ell_1 + \ell_2) \times (\ell_1 + \ell_2)$ determinant with the structure
\begin{equation}\label{2}
\rho_{\ell_1,\ell_2}(x_1,\dots,x_{\ell_1};
y_1,\dots,y_{\ell_2}) =
\det \begin{bmatrix} [ \tilde{K}_{00}(x_a, x_b) ]_{a,b=1}^{\ell_1} & 
 [\tilde{K}_{01}(x_a, y_{\beta}) ]_{a =1,\dots\ell_1 \atop \alpha = 1,\dots, \ell_2} \\
 [\tilde{K}_{10}(y_{\alpha}, x_b) ]_{\alpha =1,\dots\ell_2 \atop b = 1,\dots, \ell_1}  &
  [\tilde{K}_{11}(y_{\alpha}, y_\beta) ]_{\alpha, \beta =1}^{\ell_2} \end{bmatrix},
  \end{equation}
  for certain correlation kernels  $\tilde{K}_{00}(x,x'), \tilde{K}_{01}(x,y),  \tilde{K}_{01}(y,x),  \tilde{K}_{01}(y,y')$ dependent of $\omega_1, \omega_2,N$ but independent
  of $\ell_1, \ell_2$. The Cauchy two-matrix model is thus an example of a two-component determinantal point process; see e.g.~\cite[\S 5.9]{Fo10}.

The choice $\omega_1(x) = x^a e^{-x}$,  $\omega_2(y) = y^b e^{-y}$ restricted to $\mathbb R_+\times\mathbb{R}_+$ in (\ref{1}) is said to define the Cauchy-Laguerre 
two-matrix model. The hard edge scaled limit of this model, which is the neighbourhood of the origin on the 
positive half line of the corresponding point process with a scaling
so that eigenvalues have order unity spacing, has been analysed in \cite{bertola2014}.
It was found that the corresponding scaled correlation kernels can be written in terms of particular
Meijer-G functions; see e.g.~\cite{beals2013} for a review of the latter. Not long after, it was found that particular Meijer-G kernels also control the hard edge correlations for many examples of determinantal
point processes specified by the squared singular values of particular random matrix products 
\cite{AIK13, KZ14, Fo14,KKS15,St15,FL16}. A special case, which can be written in terms for the Wright's generalised Bessel function, had in fact been found earlier in the analysis
of the hard edge limit for what are now referred to as Muttalib--Borodin ensembles \cite{muttalib1995,Bor99}; see also the more recent work
\cite{FW15}. This latter class of models are specified by eigenvalue PDFs of the form
\begin{equation}\label{MB}
\prod_{1 \le j < k \le N} (x_k - x_j) (x_k^\theta - x_j^\theta) \prod_{j=1}^N \omega(x_j), \quad {\text{with proper weight function $\omega$}}.
\end{equation}
Note that the case $\theta = 1$ of (\ref{MB}) corresponds to the eigenvalue probability density function for a unitary invariant ensemble of
Hermitian matrices (see e.g.~\cite{PS11})--- the
Muttalib--Borodin ensembles can thus be regarded as a $\theta$-deformation of the latter.

 In \cite{bertola2009} it was conjectured that there exists a relationship between the Cauchy two-matrix and $O(1)$-model \cite{Ko89}  --- also referred
 to as the Bures ensemble \cite{SZ03}  --- for which the PDF is of the form
\begin{align}\label{Bu}
\prod_{1\leq j<k\leq N}\frac{x_k-x_j}{x_k+x_j}(x_k-x_j)\prod_{j=1}^N \omega(x),\quad {\text{with proper weight function $\omega$}}.
\end{align}
Later, one of us and Kieburg  \cite{forrester2016} gave explicit inter-relations for
the average characteristic polynomials, partition functions and correlation functions. 
A priori, this is far from obvious as the ensemble (\ref{Bu}) is an example of a Pfaffian point process (see e.g.~\cite[Ch.~6]{Fo10}), and thus has
correlations of the form
\begin{align}\label{Bu1}
\rho_\ell(x_1,\dots,x_\ell) = {\rm Pf} 
\begin{bmatrix} D(x_j,x_k) & S(x_j,x_k) \\
- S(x_k, x_j) & I(x_j,x_k) \end{bmatrix}_{j,k=1}^\ell ,
\end{align}
with the requirement that $D(x,y) = - D(y,x)$, and $I(x,y) = -I(y,x)$, so that the matrix is anti-symmetric,
whereas 
we know that (\ref{1}) is determinantal.
Using then results for the Cauchy-Laguerre two-matrix model
from  \cite{bertola2014}, the corresponding properties of the  Bures ensemble with Laguerre weight --- such as the correlation functions and
their hard edge scaled limit  --- could be made explicit.

The hard edge scaling limit is but one of a number of well defined scaling limits for (\ref{1}), (\ref{MB}) and (\ref{Bu}). Another is the global scaling limit,
corresponding to a scaling of the eigenvalues
for which in the $N \to \infty$ limit the eigenvalue density is supported on some finite interval $[0,L]$, $L > 0$.
For the Muttalib--Borodin ensemble with Laguerre weight it  is known that after the change of variables $x_l = y_l^{1/\theta}$, the global
density $\rho(y)$ minimises the energy functional \cite{CR14,FL15,FLZ15}
\begin{align}\label{E}
E_{\theta + 1}[\rho(y)] =
\theta \int_0^L y^{1/\theta} \rho(y) \, dy -
{1 \over 2} \int_0^L dy   \int_0^L dy' \,  \rho(y) \rho(y')  \log \Big ( | y^{1/\theta} - (y')^{1/\theta} | \, | y - y'| \Big )
\end{align}
with $L = \theta (1 + 1/\theta)^{1 + \theta}$. Furthermore, its moments are given by
\begin{align}\label{E1}
\int_0^L y^n \rho(y) \, dy = {1 \over \theta n + 1}
\binom{n(1 + \theta)}{n}.
\end{align}
The RHS in the case $\theta = 1$ is recognised as the $n$-th Catalan number; for general $\theta \in \mathbb Z^+$ one has instead the 
$n$-th Fuss--Catalan number with parameter $\theta + 1$. It was proposed in \cite{FL15}, and later verified in
\cite{FLZ15},
 that the Raney generalisation of the
Fuss--Catalan numbers, as specified by the sequence  
\begin{equation}\label{Ra}
R_{p,r}(n):={r \over p n + r}
\binom{p n + r}{n}, \qquad (n=0,1,2,\dots),
\end{equation}
 are realised in the case $(p,r) = (\theta/2 + 1, 1/2)$ by the moments of the minimiser of the energy functional
\begin{align}\label{E1a}
E_{\theta/2 + 1,1/2}[\rho(y)] =
\theta \int_0^L y^{1/\theta} \rho(y) \, dy -
{1 \over 2} \int_0^L dy   \int_0^L dy' \,  \rho(y) \rho(y')  \log \Big ( {| y^{1/\theta} - (y')^{1/\theta} | \over
| y^{1/\theta} + (y')^{1/\theta} |}
 \, | y - y'| \Big )
\end{align}
for certain $L$. One reads off from this the underlying probability density function 
\begin{align}\label{E2}
\prod_{1\leq j<k\leq N}\frac{x_k-x_j}{x_k+x_j}(x_k^\theta-x_j^\theta)\prod_{j=1}^N x_j^ae^{-x_j},
\end{align}
(note that (\ref{E1a}) is independent of $a$) with the change of variables $x_l = y_l^{1/\theta}$. Comparing with (\ref{Bu}) we immediately see that
(\ref{E2}) is a $\theta$-deformation of the Bures ensemble with Laguerre weight, in the same sense that the Muttalib--Borodin ensembles (\ref{MB}) are a 
 $\theta$-deformation of the unitary invariant ensembles. Moreover, as to be demonstrated below,
 it retains the property of being a Pfaffian point process.
 
 Our aim in this work is to study the statistical systems implied by (\ref{E2}) for general $\theta > 0$.
  The known relationship between the Bures ensemble and Cauchy two-matrix
 model suggests that for this purpose we introduce the $\theta$-deformation of the Cauchy-Laguerre two-matrix model, specified by the 
 PDF
\begin{align}\label{E3}
\prod_{j=1}^N x_j^ae^{-x_j}y_j^be^{-y_j}\frac{\prod_{1\leq j<k\leq N}(x_k-x_j)(y_k-y_j)}{\prod_{j,k=1}^N(x_j+y_k)}\prod_{1\leq j<k\leq N}(x_k^\theta-x_j^\theta)(y_k^\theta-y_j^\theta).
\end{align}
It will be demonstrated in Section \ref{sec:CBOPs} that this specifies a two-component determinantal point process for general $\theta > 0$.
Particular $\theta$-deformed Cauchy bi-orthogonal polynomials deduced in Section \ref{s2.1}
allow the correlation kernels in (\ref{2}) to be specified. From this the hard edge limit
can be studied, giving rise to a $\theta$ dependent family of kernels outside of the Meijer G-function class. In fact, these kernels belong to the family of Fox H-functions \cite{fox61}, which can be viewed as a generalisation of Meijer G-functions.
In section \ref{sec:bures}, analogous to the study \cite{forrester2016} in the case $\theta = 1$,
the point process corresponding to the $\theta$-deformation (\ref{E2}) is studied  through its relationship to (\ref{E3}).
We show the characteristic polynomials, partition functions of $\theta$-deformed Cauchy two-matrix model and Bures ensemble are 
related to each other, allowing in particular the computation of the correlations and their hard edge scaled limit for the latter.

\section{$\theta$-deformation of Cauchy two matrix model}\label{sec:cauchy}
Consider the $\theta$-deformation of the Cauchy-Laguerre two-matrix ensemble specified by the eigenvalue PDF (\ref{E3}).
Define the corresponding partition function 
\begin{align*}
Z_N(a,b;\theta)=
\left(\frac{1}{N!}\right)^2\int_{\mathbb{R}_+^N\times\mathbb{R}_+^N}&\prod_{j=1}^N x_j^ae^{-x_j}y_j^be^{-y_j}\frac{\prod_{1\leq j<k\leq N}(x_k-x_j)(y_k-y_j)}{\prod_{j,k=1}^N(x_j+y_k)}\\
\times&\prod_{1\leq j<k\leq N}(x_k^\theta-x_j^\theta)(y_k^\theta-y_j^\theta)d[x]d[y],
\end{align*}
where $d[x]=dx_1 \cdots dx_N$ and similarly the meaning of $d[y]$. We assume $a>-1$ and $b>-1$ to ensure the convergency of the integral.

\begin{lemma}
We have
\begin{align}\label{normalisationconstant}
Z_N(a,b;\theta)=\det\left(I_{j,k}(a,b;\theta)\right)_{j,k=1}^N,\quad I_{j,k}(a,b;\theta)=\int_{\mathbb{R}_+^2}\frac{x^{a+\theta(j-1)}y^{b+\theta(k-1)}}{x+y}e^{-(x+y)}\, dxdy.
\end{align}
\end{lemma}

\begin{proof}
We require the Cauchy double alternant identity
\begin{align}\label{CDA}
\det\left[\frac{1}{x_j+y_k}\right]_{j,k=1}^N=\frac{\prod_{1\leq j<k\leq N}(x_k-x_j)(y_k-y_j)}{\prod_{j,k=1}^N(x_j+y_k)}
\end{align}
(see e.g.~\cite{IOTZ06}) and the  Vandermonde determinant formula
\begin{align}\label{Van}
\det [z_j^{k-1}]_{j,k=1}^N = \prod_{1 \le j < k \le N} (z_k - z_j),
\end{align}
the latter used twice, with $z_j = x_j^\theta$ and with $z_j = y_j^\theta$. This allows all the double products in the integrand to be replaced by
determinants: there are three in total. Moreover, the fact that two of the determinants are anti-symmetric in $\{x_j\}$ and two are 
anti-symmetric in $\{y_j\}$ tells us that both Vandermonde determinants can be replaced by their diagonal term, provided we multiply by $(N!)^2$ ---
the reason being that each term in the sum form of the Vandermonde determinants contributes the same. These diagonal terms can be
multiplied into Cauchy double alternant determinant in such a way that only row $j$ depends on $x_j$ and column $k$ depends on
$y_k$. The integrations can then be done row-by-row (for the $x$'s) and column-by-column (for the $y$'s) to obtain 
(\ref{normalisationconstant}).
\end{proof}

\begin{remark}
This proof is very similar to that used to derive the classical Andr\'eief integration  formula; see \cite{Fo18}.
\end{remark}

The double integral in (\ref{normalisationconstant}) can be computed explicitly, and from this a product formula evaluation of the partition function
can be given.
\begin{lemma}\label{L1}
We have
\begin{align}\label{eq_i}
I_{j,k}(a,b;\theta)=\frac{\Gamma(a+\theta(j-1)+1)\Gamma(b+\theta(k-1)+1)}{1+a+b+\theta(j+k-2)}.
\end{align}
Also
\begin{align}\label{eq:partition}
Z_N(a,b;\theta)=\prod_{j=1}^N\Gamma(a+\theta(j-1)+1)\Gamma(b+\theta(j-1)+1)\times\frac{1}{\theta^N}(\prod_{l=1}^{N-1}l!)^2\prod_{k=1}^N\frac{(\beta+k-2)!}{(\beta+k+N-2)!}
\end{align}
with $\beta=(1+a+b)/\theta$.
\end{lemma}

\begin{proof}
To evaluate the double integral, introduce 
\begin{align*}
J_{j,k}(t)=\int_{\mathbb{R}_+^2}\frac{x^{a+\theta(j-1)}y^{b+\theta(k-1)}}{x+y}e^{-t(x+y)} \,dxdy.
\end{align*}
Through the variable transformation $x\mapsto x/t$ and $y\mapsto y/t$, we see that the dependence on $t$ can be written as $J_{j,k}(t)=t^{-(1+a+b+\theta(j+k-2))}I_{j,k}(a,b;\theta)$ and
\begin{align*}
\frac{d}{dt}J_{j,k}(t)=-t^{-(2+a+b+\theta(j+k-2))}(1+a+b+\theta(j+k-2))I_{j,k}(a,b;\theta).
\end{align*}
On the other hand, if we take the derivative of $J_{j,k}(t)$, we can obtain
\begin{align*}
\frac{d}{dt}J_{j,k}(t)&=-\left(\int_{\mathbb{R}_+}x^{a+\theta(j-1)}e^{-tx} \,dx\right)\left(\int_{\mathbb{R}_+}y^{b+\theta(k-1)}e^{-ty} \, dy\right)\\&=-t^{-(2+a+b+\theta(j+k-2))}\Gamma(a+\theta(j-1)+1)\Gamma(b+\theta(k-1)+1).
\end{align*}
Comparing these two formulae, it gives (\ref{eq_i}).

Substituting  (\ref{eq_i}) into \eqref{normalisationconstant}, we have
\begin{align*}
Z_N(a,b;\theta)=\prod_{j=1}^N\Gamma(a+\theta(j-1)+1)\Gamma(b+\theta(j-1)+1)\times\det\left[\frac{1}{1+a+b+\theta(j+k-2)}\right]_{j,k=1}^N.
\end{align*}
This determinant is of the Cauchy double alternant type, and so according to (\ref{CDA}) has the evaluation
\begin{align*}
\det\left[\frac{1}{1+a+b+\theta(j+k-2)}\right]_{j,k=1}^N=\frac{1}{\theta^N}(\prod_{l=1}^{N-1}l!)^2\prod_{k=1}^N\frac{(\beta+k-2)!}{(\beta+k+N-2)!}
\end{align*}
if one denotes $\beta=\frac{1+a+b}{\theta}$. The result (\ref{eq:partition}) now follows.
\end{proof}

\begin{remark}
When $\theta=1$, this agrees with known results; see \cite{bertola2009}.
\end{remark}

\subsection{Jacobi polynomials and $\theta$-deformation Cauchy bi-orthogonal polynomials}
Associated with the double integral in  (\ref{normalisationconstant}) is an inner product, and associated with the inner
product are biorthogonal polynomials. 

\begin{define}
Under the inner product
\begin{align*}
\mu_{i,j}=\langle x^i, y^j\rangle=\iint_{\mathbb{R}_+^2}\frac{x^{i\theta+a}y^{j\theta+b}}{x+y}e^{-(x+y)}dxdy,
\end{align*}
one can define a family of $\theta$-deformed Cauchy bi-orthogonal polynomials $\{\hat{P}_n(x),\hat{Q}_n(y)\}_{n\in\mathbb{N}}$ with Laguerre weight
having the property
\begin{align}\label{PQ1}
\langle \hat{P}_n(x),\hat{Q}_m(y)\rangle=\iint_{\mathbb{R}_+^2}\frac{e^{-(x+y)}}{x+y}x^ay^b\hat{P}_m(x^\theta)\hat{Q}_n(y^\theta)
\, dxdy=\frac{h_n}{\theta}\delta_{n,m},
\end{align}
where $h_n = {\theta \over 2n\theta+ a+b + 1}$ (this choice of normalisation is for latter convenience).
\end{define}

In the case $\theta = 1$, explicit forms for $\{\hat{P}_n(x),\hat{Q}_n(y)\}_{n\in\mathbb{N}}$ have been obtained in \cite{bertola2009}.
The method of derivation can be generalised to give the corresponding explicit forms for all $\theta > 0$.

\begin{proposition}
Set
\begin{align*}
c_{n,l}=\frac{\Gamma(\alpha+n+l+1)}{l!(n-l)!\Gamma(\alpha+l+1)}(-1)^l.
\end{align*}
The polynomials
\begin{align}\label{PQ}
\hat{P}_n(x)=\sum_{l=0}^n c_{n,l}\frac{x^l}{\Gamma(a+\theta l+1)},\quad \hat{Q}_n(y)=\sum_{l=0}^n c_{n,l}\frac{y^l}{\Gamma(b+\theta l+1)},
\end{align}
are $\theta$-deformed Cauchy bi-orthogonal polynomials. They can be written in terms of contour integrals according to
\begin{subequations}
\begin{align}
\hat{P}_n(x) & =   \int_{\gamma}\frac{du}{2\pi i}\frac{\Gamma(\alpha+n-u+1)\Gamma(u)}{\Gamma(n+u+1)\Gamma(\alpha-u+1)\Gamma(a-\theta u+1)}x^{-u},\label{eq:pint} \\
\hat{Q}_n(y) & =\int_{\gamma}\frac{du}{2\pi i}\frac{\Gamma(\alpha+n-u+1)\Gamma(u)}{\Gamma(n+u+1)\Gamma(\alpha-u+1)\Gamma(b-\theta u+1)}y^{-u}\label{eq:qint}
\end{align}
\end{subequations}
valid for $x,y \ne 0$ and the contour encloses $\{0,-1,\dots,-n\}$. 

\end{proposition}

\begin{proof}

With the above definition of $c_{n,l}$ we note that
\begin{align}\label{Pa}
P^{(\alpha,0)}_n(1-2x):=P^{((a+b+1)/\theta-1,0)}_n(1-2x):=\sum_{i=0}^n c_{n,i}x^i,
\end{align}
where $P^{(\alpha,0)}_n$ denotes a Jacobi polynomial with parameters $(\alpha,0)$.
Also, upon
recognising that
\begin{align}\label{AB1}
\frac{1}{1+a+b+\theta(j+k-2)}=\frac{1}{\theta}\frac{1}{(1+a+b)/\theta+j+k-2}=\frac{1}{\theta}\int_0^1 x^{j-1}x^{k-1}x^{\frac{1+a+b}{\theta}-1}dx,
\end{align}
we see that $I_{j,k}(a,b;\theta)$ as defined in (\ref{normalisationconstant}) and evaluated by (\ref{eq_i}) can be written in the integral form
\begin{align}\label{IA}
I_{j,k}(a,b;\theta) = \frac{\Gamma(a+\theta(j-1)+1)\Gamma(b+\theta(k-1)+1)}{\theta} \int_0^1 x^{j-1}x^{k-1}x^{\frac{1+a+b}{\theta}-1}dx.
\end{align}

To make use of first (\ref{IA}), we note from (\ref{PQ1}) with the substitution (\ref{PQ}) that
\begin{align*}
\langle \hat{P}_n(x),\hat{Q}_m(y)\rangle=
\sum_{j=0}^n \frac{c_{n,j}}{\Gamma(a+\theta j+1)}
\sum_{k=0}^m \frac{c_{m,k}}{\Gamma(b+\theta k+1)} I_{j,k}(a,b;\theta),
\end{align*}
Rewriting according to (\ref{IA}) we have
\begin{align*}
\langle \hat{P}_n(x),\hat{Q}_m(y)\rangle= {1 \over \theta}
\sum_{j=0}^n  c_{n,j}
\sum_{k=0}^m  c_{m,k}  \int_0^1 x^{j-1}x^{k-1}x^{\frac{1+a+b}{\theta}-1}dx.
\end{align*}
Using now (\ref{Pa}) gives
\begin{align*}
\langle \hat{P}_n(x),\hat{Q}_m(y)\rangle= {1 \over \theta}
\int_0^1 x^\alpha  P^{(\alpha,0)}_n(1-2x) P^{(\alpha,0)}_m(1-2x) \, dx.
\end{align*}
The orthogonality of the Jacobi polynomials $\int_0^1 x^\alpha  P^{(\alpha,0)}_n(1-2x) P^{(\alpha,0)}_m(1-2x) \, dx = h_n \delta_{n,m}$ with $h_n=\frac{1}{2n+\alpha+1}$ gives (\ref{PQ1}) as required.

Regarding the contour integrals, note that for the poles of $\Gamma(u)$ for $u<-n-1$ in the integrand
 are cancelled by the poles of the $\Gamma(n+1+u)$ and thus this integral is a polynomial in $x$ of degree $n$.
  Calculation of the residues reclaims (\ref{PQ}).
\end{proof}
\begin{remark} We can check, using integration methods analogous to the one used in Lemma \ref{L1}, that
the $\theta$-deformed Cauchy bi-orthogonal polynomials can be written in the determinant form
\begin{align*}
&\hat{P}_n(x)=\sqrt{\frac{h_n}{\theta Z_nZ_{n+1}}}\left|
\begin{array}{cccc}
\mu_{0,0}&\cdots&\mu_{0,n-1}&1\\
\mu_{1,0}&\cdots&\mu_{1,n-1}&x\\
\vdots&&\vdots&\vdots\\
\mu_{n,0}&\cdots&\mu_{n,n-1}&x^n
\end{array}
\right|,\\
& \hat{Q}_n(y)=\sqrt{\frac{h_n}{\theta Z_nZ_{n+1}}}\left|
\begin{array}{cccc}
\mu_{0,0}&\mu_{0,1}&\cdots&\mu_{0,n}\\
\vdots&\vdots&&\vdots\\
\mu_{n-1,0}&\mu_{n-1,1}&\cdots&\mu_{n-1,n}\\
1&y&\cdots&y^n
\end{array}
\right|.
\end{align*}
\end{remark}

Interestingly, the integrals \eqref{eq:pint} and \eqref{eq:qint} can be written in terms of the Fox H-function. In general, the Fox H-function is a generalisation of the Meijer G-function \cite{fox61}. It is defined by a Mellin-Barnes integral
\begin{align}\label{eq:foxh}
\begin{aligned}
&H_{p,q}^{\,m,n} \!\left[ z \left| \begin{matrix}
( a_1 , A_1 );~( a_2 , A_2 );~\ldots;~( a_p , A_p ) \\
( b_1 , B_1 );~( b_2 , B_2 );~\ldots;~( b_q , B_q ) \end{matrix} \right. \right]\\
&\qquad\quad\quad:=\frac{1}{2\pi i}\int_\gamma\frac{(\prod_{j=1}^m\Gamma(b_j+B_js))(\prod_{j=1}^n\Gamma(1-a_j-A_js))}{(\prod_{j=m+1}^q\Gamma(1-b_j-B_js))(\prod_{j=n+1}^p\Gamma(a_j+A_js))}z^{-s}ds,
\end{aligned}
\end{align}
for a suitable contour $\gamma$.
Therefore, the bi-orthogonal polynomials $\hat{P}_n(x)$ and $\hat{Q}_n(x)$ can be expressed as
\begin{align*}
\hat{P}_n(x)=H_{2,3}^{1,1}\!\left[x\left|\begin{aligned}
(-\alpha-n,1);~(n+1,1)\\
(0,1);~(-\alpha,1);~(-a,\theta)
\end{aligned}
\right.\right],\quad \hat{Q}_n(y)=H_{2,3}^{1,1}\!\left[y\left|\begin{aligned}
(-\alpha-n,1);~(n+1,1)\\
(0,1);~(-\alpha,1);~(-b,\theta)
\end{aligned}
\right.\right].
\end{align*}
 
\subsection{Christoffel-Darboux kernel for the $\theta$-deformed Cauchy bi-orthogonal polynomials}\label{s2.1}
In this section, we introduce the Christoffel-Darboux kernel for the $\theta$-deformed Cauchy bi-orthogonal polynomials. 
In the next section it will be shown to play a key role in the specification of the
 correlation kernel for the $\theta$-deformation Cauchy two-matrix model.

\begin{define}
The Christoffel-Darboux kernel for the $\theta$-deformation Cauchy bi-orthogonal polynomials is defined as
\begin{align}\label{eq:CDkernel}
K_N(x,y)=\sum_{n=0}^{N-1}\frac{\theta}{h_n}\hat{P}_n(x^\theta)\hat{Q}_n(y^\theta).
\end{align}
\end{define}

It is not difficult to see that this is a reproducing kernel, as specified by certain integration formulas.

\begin{lemma}
We have
\begin{align*}
&(1)\quad\int_{\mathbb{R}_+^2}K_N(x,z)K_N(w,y)\frac{w^ae^{-w}z^be^{-z}}{z+w} \, dzdw=K_N(x,y);\\
&(2)\quad\int_{\mathbb{R}_+^2}K_N(x,y)\frac{x^ae^{-x}y^be^{-y}}{x+y} \,dxdy=N.
\end{align*}
\end{lemma}

\begin{proof}
The proof is direct. Noting that
\begin{align*}
\int_{\mathbb{R}_+^2}&K_N(x,z)K_N(w,y)\frac{w^ae^{-w}z^be^{-z}}{w+z}dzdw\\
&=\int_{\mathbb{R}_+^2}\sum_{j=0}^{N-1}\frac{\theta}{h_j}\hat{P}_j(x^\theta)\hat{Q}_j(z^\theta)\sum_{i=0}^{N-1}\frac{\theta}{h_i}\hat{P}_i(w^\theta)\hat{Q}_j(y^\theta)\frac{w^ae^{-w}z^be^{-z}}{w+z}\, dzdw\\
&=\sum_{j=0}^{N-1}\sum_{i=0}^{N-1}\frac{\theta}{h_j}\frac{\theta}{h_i}\hat{P}_j(x^\theta)\hat{Q}_i(y^\theta)\frac{h_j}{\theta}\delta_{i,j}=\sum_{j=0}^{N-1}\frac{\theta}{h_j}\hat{P}_j(x^\theta)\hat{Q}_j(y^\theta)=K_N(x,y),
\end{align*}
the first equality is verified. Then we turn to the second equality,
\begin{align*}
\int_{\mathbb{R}_+^2}K_N(x,y)\frac{x^ae^{-x}y^be^{-y}}{x+y}\, dxdy=\int_{\mathbb{R}_+^2}\sum_{j=0}^{N-1}\frac{\theta}{h_j}\hat{P}_j(x^\theta)\hat{Q}_j(y^\theta)\frac{x^ae^{-x}y^be^{-y}}{x+y}\, dxdy=\sum_{j=0}^{N-1}1=N.
\end{align*}
\end{proof}

It furthermore permits integral forms.
\begin{proposition}\label{P210}
Let 
\begin{equation}\label{ap}
\alpha = (a+b+1)/\theta - 1
\end{equation}
as in (\ref{Pa}).
We have
\begin{align}\label{eq:kernel1}
\begin{aligned}
K_N(x,y)&=\theta\int_{\gamma\times\gamma}\frac{du}{2\pi i}\frac{dv}{2\pi i}\frac{x^{-u\theta}y^{-v\theta}}{1+\alpha-u-v}\\
&\times
\frac{\Gamma(u)\Gamma(v)\Gamma(\alpha+N+1-u)\Gamma(\alpha+N+1-v)}{\Gamma(N+u)\Gamma(N+v)\Gamma(\alpha+1-u)\Gamma(\alpha+1-v)}\frac{1}{\Gamma(a-\theta u+1)\Gamma(b-\theta v+1)}.
\end{aligned}
\end{align}
Also, setting
\begin{align*}
G_{N,a}^\theta(z)=\int_{\gamma}\frac{du}{2\pi i}\frac{\Gamma(u)\Gamma(\alpha+N+1-u)}{\Gamma(N+u)\Gamma(\alpha+1-u)\Gamma(a-\theta u+1)}z^{-u},
\end{align*}
we have
\begin{align*}
K_N(x,y)=\theta\int_0^1 t^\alpha G_{N,a}^\theta(tx^\theta)G_{N,b}^\theta(ty^\theta)dt.
\end{align*}
\end{proposition}

\begin{proof}
Since $\hat{P}_n(x)$ and $\hat{Q}_n(y)$ have integral formula \eqref{eq:pint} and \eqref{eq:qint}, we can equivalently express $K_N(x,y)$ as
\begin{align*}
K_N(x,y)&=\sum_{n=0}^{N-1}\frac{\theta}{h_n}\hat{P}_n(x^\theta)\hat{Q}_n(y^\theta)\\
&=\theta\int_{\gamma\times\gamma}\frac{du}{2\pi i}\frac{dv}{2\pi i}\frac{\Gamma(u)x^{-\theta u}}{\Gamma(\alpha+1-u)\Gamma(a-\theta u+1)}\frac{\Gamma(v)y^{-\theta v}}{\Gamma(\alpha+1-v)\Gamma(b-\theta v+1)}\\
&\times\sum_{n=0}^{N-1}(2n+\alpha+1)\left(\frac{\Gamma(\alpha+n-u+1)\Gamma(\alpha+n-v+1)}{\Gamma(n+1+u)\Gamma(n+1+v)}\right).
\end{align*}
Proceeding as in \cite{KZ14}, by
employing the formula
\begin{align*}
&\sum_{n=0}^{N-1}(2n+\alpha+1)\frac{\Gamma(\alpha+n-u+1)\Gamma(\alpha+n-v+1)}{\Gamma(n+1+u)\Gamma(n+1+v)}\\
&=\frac{1}{1-\alpha-u-v}\left(\frac{\Gamma(\alpha+N+1-u)\Gamma(\alpha+N+1-v)}{\Gamma(N+u)\Gamma(N+v)}-\frac{\Gamma(\alpha+1-u)\Gamma(\alpha+1-v)}{\Gamma(u)\Gamma(v)}\right)
\end{align*}
to the summation in the kernel and noting that 
\begin{align*}
\int_{\gamma\times\gamma}\frac{du}{2\pi i}\frac{dv}{2\pi i}\frac{x^{-\theta u}y^{-\theta v}}{1+\alpha-u-v}=0,\quad \text{for~} \alpha>-1
\end{align*}
gives (\ref{eq:kernel1}). The formula
 \begin{equation}\label{U}
 \frac{x^{-u\theta}y^{-v\theta}}{1+\alpha-u-v}=\int_0^1 t^\alpha(tx^\theta)^{-u}(ty^\theta)^{-v}dt
 \end{equation}
 now gives the second form.
 \end{proof}
 
Finally we would like to analyse the hard edge scaling of  the Christoffel-Darboux kernel.

\begin{corollary}
We have
\begin{align*}
\lim_{N\to\infty} {N^{-2(\alpha+1)}}K_N\Big ( \frac{X}{N^{\frac{2}{\theta}}},\frac{Y}{N^{\frac{2}{\theta}}} \Big )=\theta\int_0^1 t^\alpha G_{\infty,a}^\theta(tX^\theta)G_{\infty,b}^\theta(tY^\theta)dt,
\end{align*}
where
\begin{align*}
G_{\infty,a}^\theta(z)=\int_\gamma \frac{du}{2\pi i} \frac{\Gamma(u)z^{-u}}{\Gamma(\alpha-u+1)\Gamma(a-\theta u+1)} =
\sum_{k=0}^\infty {z^k \over k! \Gamma(\alpha + k + 1) \Gamma(a+\theta k + 1)},
\end{align*}
and the contour $\gamma$ is required to enclose the negative real axis, and the origin.
\end{corollary}

\begin{proof}
This integral form follows from Proposition \ref{P210} by applying the uniform asymptotic estimate
\begin{align}\label{eq:asymptotic}
\frac{\Gamma(\alpha+N+1-u)\Gamma(\alpha+N+1-v)}{\Gamma(N+u)\Gamma(N+v)}\sim N^{\alpha+1-2u}N^{\alpha+1-2v},\quad N\to\infty
\end{align}
and the series form follows from the integral by evaluating the residues.
\end{proof}

The Christoffel-Darboux kernel and its asymptotic estimation all belong to the Fox H-function class since
\begin{align}\label{tginf}
\begin{aligned}
G_{N,a}^\theta(z)=H_{2,3}^{1,1}\left[
z\left|\begin{matrix}
(-\alpha-N,1);~(N,1)\\
(0,1);~(-\alpha,1);~(-a,\theta)
\end{matrix}\right.
\right], \quad G_{\infty,a}^\theta(z)=H_{0,3}^{1,0}\left[
z\left|\begin{matrix}
-\\
(0,1);~(-\alpha,1);~(-a,\theta)
\end{matrix}
\right.
\right].
\end{aligned}
\end{align}

\subsection{Correlation functions for the $\theta$-deformation Cauchy two-matrix model}\label{sec:CBOPs}
The $\theta$-deformation Cauchy two-matrix model is a determinantal point process.
We will see that  the Christoffel-Darboux kernel underlies the expression for the corresponding correlation kernel.

\begin{define}\label{def:correlation}
The $(r,s)$-correlation function is defined as
\begin{align*}
\rho_{r,s}(x_1,\cdots,x_r&;y_1,\cdots,y_s):=\frac{N!\prod_{j=1}^r x_j^a e^{- x_j} \prod_{j=1}^s  y_j^b e^{-y_j}}{(N-r)!(N-s)!Z_N}\\
&\times\int_{\mathbb{R}_+^{N-r}\times\mathbb{R}_+^{N-s}}\prod_{j=r+1}^N\  x_j^a e^{- x_j} \, dx_j\prod_{j=s+1}^N  y_j^b e^{-y_j} \, dy_j\Delta(x^\theta)\Delta(y^\theta)\det\left[\frac{1}{x_i+y_j}\right]_{i,j=1}^N.
\end{align*}
\end{define}

\begin{lemma}
The $(r,s)$-correlation function can be re-expressed in terms of Christoffel-Darboux kernel
\begin{align*}
\rho_{r,s}(x_1,\cdots,x_r&;y_1,\cdots,y_s)=\frac{N!\prod_{j=1}^r x_j^a e^{- x_j}  \prod_{j=1}^s y_j^be^{-y_j}}{(N-r)!(N-s)!}\times\\
&\int_{\mathbb{R}_+^{N-r}\times\mathbb{R}_+^{N-s}}\prod_{l=r+1}^Nx_j^ae^{-x_j}dx_j\prod_{j=s+1}^Ny_j^be^{-y_j}dy_j\det\left[K_N(x_i,y_j)\right]_{i,j=1}^N\det\left[\frac{1}{x_i+y_j}\right]_{i,j=1}^N.
\end{align*}
\end{lemma}

\begin{proof}
First note that by making use of (\ref{Van}) we can write
\begin{align*}
\Delta_N(x^\theta)\Delta_N(y^\theta)=\left[\det(\tilde{P}_{i-1}(x_j^\theta))\det(\tilde{Q}_{i-1}(y_j^\theta))\right]_{i,j=1}^{N},
\end{align*}
where $\tilde{P}_i(x)$ and $\tilde{Q}_i(y)$ are the monic $\theta$-deformation Cauchy bi-orthogonal polynomials with the bi-orthogonal relation
\begin{align*}
\langle \tilde{P}_i(x),\tilde{Q}_j(y)\rangle=\int_{\mathbb{R}_+^2}\frac{\tilde{P}_i(x^\theta)\tilde{Q}_j(y^\theta)}{x+y}x^ay^be^{-(x+y)}dxdy=\frac{Z_{i+1}}{Z_i}\delta_{i,j}.
\end{align*}
The latter are associated with the $\hat{P}_i(x)$ and $\hat{Q}_i(x)$ in \eqref{PQ1} by the relation
\begin{align*}
\hat{P}_i(x)=\sqrt{\frac{h_i}{\theta}}\sqrt{\frac{Z_i}{Z_{i+1}}}\tilde{P}_i(x),\quad\hat{Q}_i(x)=\sqrt{\frac{h_i}{\theta}}\sqrt{\frac{Z_i}{Z_{i+1}}}\tilde{Q}_i(x).
\end{align*}
Then one can deduce
\begin{align*}
\frac{1}{Z_N}\left[\det(\tilde{P}_{i-1}(x_j^\theta))\det(\tilde{Q}_{i-1}(y_j^\theta))\right]_{i,j=1}^{N}&=\left(\prod_{i=0}^{N-1}\frac{h_i}{\theta}\right)\left[\det(\hat{P}_{i-1}(x_j^\theta))\det(\hat{Q}_{i-1}(y_j^\theta))\right]_{i,j=1}^{N}\\&=\det(K_N(x_i,y_j))_{i,j=1}^{N}
\end{align*}
with the reproducing kernel $K_N(x,y)$ defined in \eqref{eq:CDkernel}, which implies the result.
\end{proof}

In the special case $r=s=N$, the above re-expressed kernel reads
\begin{align*}
\rho_{N,N}(x_1,\cdots,x_N;y_1,\cdots,y_N)=\frac{1}{N!}\det\left[K_N(x_i,y_j)\right]_{i,j=1}^N\det\left[\frac{1}{x_i+y_j}\right]_{i,j=1}^N\prod_{j=1}^Nx_j^ay_j^be^{-(x_j+y_j)}dx_jdy_j.
\end{align*}
We will show that that this can be expressed in terms of correlation kernels.

\begin{proposition}
One has
\begin{align*}
\det\left[K_N(x_i,y_j)\right]_{i,j=1}^N\det\left[\frac{1}{x_i+y_j}\right]_{i,j=1}^N=\det\left[
\begin{array}{cc}
\left[K_{01}(x_i,x_j)\right]_{i,j=1}^N&\left[K_{00}(x_i,y_j)\right]_{i,j=1}^N\\
\left[K_{11}(y_i,x_j)\right]_{i,j=1}^N&\left[K_{10}(y_i,y_j)\right]_{i,j=1}^N
\end{array}
\right],
\end{align*}
where $K_{00}(x,y)$ is exactly the Christoffel-Darboux kernel $K_N(x,y)$ and $K_{01}$, $K_{10}$ and $K_{11}$ are defined as
\begin{align}\label{eq:integralkernel}
\begin{aligned}
&K_{01}(x,x'):=\int_{\mathbb{R}_+}\frac{K_N(x,y)}{x'+y}y^be^{-y}dy,\\
&K_{10}(y,y'):=\int_{\mathbb{R}_+}\frac{K_N(x,y')}{x+y}x^ae^{-x}dx,\\
&K_{11}(y,x):=\int_{\mathbb{R}_+^2}\frac{K_N(x',y')}{(x'+y)(x+y')}x'^ay'^be^{-(x'+y')}dx'dy'-\frac{1}{x+y}.
\end{aligned}
\end{align}
\end{proposition}

\begin{proof}
This proposition is directly proved by the matrix decomposition. By denoting 
\begin{align*}
&P_i(x)=\sqrt{\frac{h_i}{\theta}}\hat{P}_i(x),\quad P_i^{(1)}(y)=\int_{\mathbb{R}_+}\frac{P_i(x^\theta)}{x+y}x^ae^{-x}dx,\\
&Q_i(x)=\sqrt{\frac{h_i}{\theta}}\hat{Q}_i(x),\quad Q_i^{(1)}(x)=\int_{\mathbb{R}_+}\frac{Q_i(y^\theta)}{x+y}y^be^{-y}dy,
\end{align*}
we know the matrix in right hand side can be decomposed as
\begin{align*}
&\left[
\begin{array}{cc}
\left[K_{01}(x_i,x_j)\right]_{i,j=1}^N&\left[K_{00}(x_i,y_j)\right]_{i,j=1}^N\\
\left[K_{11}(y_i,x_j)\right]_{i,j=1}^N&\left[K_{10}(y_i,y_j)\right]_{i,j=1}^N
\end{array}
\right]\\
&=\left[\begin{array}{cc}
\left({P}_{i-1}(x_j^\theta)\right)_{i,j=1}^N&0\\
\left(P^{(1)}_{i-1}(y_j)\right)_{i,j=1}^N&I
\end{array}
\right]\left[
\begin{array}{cc}
0&I\\
I&0
\end{array}
\right]\left[
\begin{array}{cc}
-\left(\frac{1}{x_i+y_j}\right)_{i,j=1}^N&0\\
\left(Q_{i-1}^{(1)}(x_j)\right)_{i,j=1}^N&\left(Q_{i-1}(y^\theta_j)\right)_{i,j=1}^N
\end{array}
\right],
\end{align*}
and hence the determinant of this matrix is equal to $\det\left[K_N(x_i,y_j)\right]_{i,j=1}^N\det\left[\frac{1}{x_i+y_j}\right]_{i,j=1}^N$.
\end{proof}
By employing the above, one can see the $(N,N)$-correlation function can be written as 
\begin{align*}
\rho_{N,N}(x_1,\cdots,x_N;y_1,\cdots,y_N)=\prod_{i=1}^Nx_i^ay_i^be^{-(x_i+y_i)}\det\left[
\begin{array}{cc}
\left[K_{01}(x_i,x_j)\right]_{i,j=1}^N&\left[K_{00}(x_i,y_j)\right]_{i,j=1}^N\\
\left[K_{11}(y_i,x_j)\right]_{i,j=1}^N&\left[K_{10}(y_i,y_j)\right]_{i,j=1}^N
\end{array}
\right].
\end{align*}
This structure remains true for the general $(r,s)$-correlation functions.
\begin{proposition}\label{prop:correlation}
The $(r,s)$-correlation functions for the $\theta$-deformation Cauchy two-matrix model is exactly
\begin{align}\label{eq:correlation1}
\begin{aligned}
\rho_{r,s}&(x_1,\cdots,x_r;y_1,\cdots,y_s)\\
&=\prod_{i=1}^r\prod_{j=1}^sx_i^ay_j^be^{-(x_i+y_j)}\times\det\left[
\begin{array}{cc}
\left[K_{01}(x_i,x_j)\right]_{i,j=1}^r&\left[K_{00}(x_i,y_j)\right]_{1\leq i\leq r}^{1\leq j\leq s}\\
\left[K_{11}(y_i,x_j)\right]_{1\leq i\leq s}^{1\leq j\leq r}&\left[K_{10}(y_i,y_j)\right]_{i,j=1}^s
\end{array}
\right].
\end{aligned}
\end{align}
\end{proposition}

\begin{proof}
This proposition is a special case of \cite[Corollary 1.4]{rains2000} if we take $\kappa(x,y)=\frac{1}{x+y}$,
$\phi_j(x) = \hat{P}_{j-1}(x^\theta)$ and $\psi_j(x) = \hat{Q}_{j-1}(x^\theta)$.
\end{proof}

\begin{remark}
In fact, one can scale the correlation function \eqref{eq:correlation1} so that the multiple factor can be absorbed into the kernel, that is
\begin{align*}
\rho_{r,s}&(x_1,\cdots,x_r;y_1,\cdots,y_s)=\det\left[
\begin{array}{cc}
\left[\hat{K}_{01}(x_i,x_j)\right]_{i,j=1}^r&\left[\hat{K}_{00}(x_i,y_j)\right]_{1\leq i\leq r}^{1\leq j\leq s}\\
\left[\hat{K}_{11}(y_i,x_j)\right]_{1\leq i\leq s}^{1\leq j\leq r}&\left[\hat{K}_{10}(y_i,y_j)\right]_{i,j=1}^s
\end{array}
\right],\end{align*}
where $\hat{K}_{00}(x,y)=K_{00}(x,y)$, $\hat{K}_{01}(x,x')=e^{-x'}x'^aK_{01}(x,x')$, $\hat{K}_{10}(y,y')=e^{-y}y^bK_{10}(y,y')$ and $\hat{K}_{11}(y,x)=e^{-(x+y)}x^ay^bK_{11}(y,x)$. This form is useful for purposes of obtaining the kernel for the Bures ensemble.
\end{remark}

Particular integral forms can be derived for the correlation kernels.
\begin{proposition}
Define
\begin{align*}
&G_{N,a}^\theta(z)=\int_{\gamma}\frac{du}{2\pi i}\frac{\Gamma(u)\Gamma(\alpha+N+1-u)}{\Gamma(N+u)\Gamma(\alpha+1-u)\Gamma(a-\theta u+1)}z^{-u},\\
&\tilde{G}_{N,a}^\theta(z)=\int_{\gamma}\frac{du}{2\pi i}\frac{\pi}{\sin(u\theta-a)\pi}\frac{\Gamma(u)\Gamma(\alpha+N+1-u)}{\Gamma(N+u)\Gamma(\alpha+1-u)\Gamma(a-\theta u+1)}z^{-u}\\
&\qquad\quad=\int_{\gamma}\frac{du}{2\pi i}\frac{\Gamma(u)\Gamma(\alpha+N+1-u)\Gamma(\theta u-a)}{\Gamma(N+u)\Gamma(\alpha+1-u)}z^{-u}.
\end{align*}
We have
\begin{align*}
&K_{01}(x,x')=\theta e^{x'}x'^b\int_0^1 t^\alpha G_{N,a}^\theta(tx^\theta)\tilde{G}_{N,b}^\theta(tx'^\theta)dt,\\
&K_{10}(y,y')=\theta e^y y^a\int_{0}^1t^\alpha\tilde{G}_{N,a}^\theta(ty^\theta)G_{N,b}^\theta(ty'^\theta)dt,\\
&K_{11}(y,x)=\theta e^{x+y} x^by^a\int_0^1 t^\alpha\tilde{G}^\theta_{N,a}(ty^\theta)\tilde{G}_{N,b}^\theta(tx^\theta)dt-\frac{1}{x+y},
\end{align*}
and the corresponding hard edge scaled limits are given by
\begin{align*}
&\lim_{N\to\infty}N^{-2(\alpha+1)+\frac{2b}{\theta}}K_{01}\Big (\frac{X}{N^{\frac{2}{\theta}}},\frac{X'}{N^{\frac{2}{\theta}}}\Big )=\theta X'^b\int_0^1 t^\alpha G_{\infty,a}^\theta(tX^\theta)\tilde{G}_{\infty,b}^{\theta}(tX'^\theta)dt,\\
&\lim_{N\to\infty}N^{-2(\alpha+1)+\frac{2a}{\theta}}K_{10}\Big (\frac{Y}{N^{\frac{2}{\theta}}},\frac{Y'}{N^{\frac{2}{\theta}}}\Big )=\theta Y^a\int_0^1t^\alpha \tilde{G}_{\infty,a}^\theta (tY^\theta)G_{\infty,b}^\theta(tY'^\theta)dt,\\
&\lim_{N\to\infty}N^{-\frac{2}{\theta}}K_{11}\Big (\frac{Y}{N^{\frac{2}{\theta}}},\frac{X}{N^{\frac{2}{\theta}}} \Big )+\frac{N^{2\theta(\alpha+1)}}{X+Y}=\theta X^bY^a\int_0^1 t^\alpha \tilde{G}_{\infty,a}(tY^\theta)\tilde{G}_{\infty,b}^\theta(tX^\theta)dt,
\end{align*}
where 
\begin{align}\label{eq:kernel}
\begin{aligned}
&G_{\infty,a}^\theta(z)=\int_{\gamma}\frac{\Gamma(u)z^{-u}}{\Gamma(\alpha+1-u)\Gamma(\alpha-\theta u+1)}\frac{du}{2\pi i},\\
&\tilde{G}_{\infty,a}^{\theta}(z)=\int_\gamma \frac{\Gamma(u)\Gamma(\theta u-a)z^{-u}}{\Gamma(\alpha+1-u)}\frac{du}{2\pi i}.
\end{aligned}
\end{align}

\end{proposition}
\begin{remark}
Here, $\tilde{G}_{N,a}^\theta(z)$ can be written in terms of a Fox H-function
\begin{align}
\tilde{G}_{N,a}^\theta(z)=H_{2,3}^{2,1}\left[
z\left|\begin{matrix}
(-\alpha-N,1);~(N,1)\\
(0,1);~(-a,\theta);~(-\alpha,1)\end{matrix}
\right.
\right],
\end{align}
and its scaled form can also be expressed as a Fox H-function
\begin{align}\label{ginf}
\tilde{G}_{\infty,a}^\theta(z)=H_{0,3}^{2,0}\left[
z\left|\begin{matrix}
-\\
(0,1);(-a,\theta);(-\alpha,1)
\end{matrix}\right.
\right].
\end{align}
\end{remark}
\begin{proof}
We give the details for $K_{01}(x,x')$ only; the derivation in the other cases is similar.
Use of the integral representation \eqref{eq:kernel1} gives
\begin{align*}
K_{01}(x,x')&=\int_{\mathbb{R}_+}\frac{K_N(x,y)}{x'+y}y^be^{-y}dy\\
&=\theta\int_{\gamma^2}\frac{dudv}{(2\pi i)^2}\frac{x^{-u\theta}}{1+\alpha-u-v}\frac{\Gamma(u)\Gamma(v)\Gamma(\alpha+N+1-u)\Gamma(\alpha+N+1-v)}{\Gamma(N+u)\Gamma(N+v)\Gamma(\alpha+1-u)\Gamma(\alpha+1-v)}\\&\times\frac{1}{\Gamma(a+1-\theta u)\Gamma(b+1-\theta v)}\int_{\mathbb{R}_+}\frac{y^{-\theta v+b}e^{-y}}{x'+y}dy,
\end{align*}
which means we only need to compute the integral $\int_{\mathbb{R}_+}\frac{y^{-v \theta+b}e^{-y}}{x'+y}dy$. One can check that
\begin{align*}
\int_0^\infty&\frac{y^{-\theta v+b}e^{-y}}{x'+y}dy=\int_0^\infty y^{-\theta v+b}e^{-y}\left(\int_0^\infty e^{(-x'+y)s}ds\right)dy\\
&=\int_0^\infty\left(\int_0^\infty y^{-\theta v+b}e^{-(s+1)y}dy\right)e^{-x'}ds
=\Gamma(1+b-\theta v)\int_0^\infty (1+s)^{\theta v-b+1}e^{-x's}ds\\
&=e^{x'}\Gamma(1+b-\theta v)\int_1^\infty t^{\theta v-b+1}e^{-tx'}dx
=e^{x'}x'^{b-v\theta}\Gamma(v\theta-b;x')\Gamma(1+b-\theta v)
\end{align*}
with $\Gamma(v\theta-b;x')$ the incomplete Gamma function. By using the relation
\begin{align*}
\Gamma(v\theta-b;x')=\Gamma(v\theta-b)-\frac{x'^{v\theta-b}}{b-v\theta}{_{1}F_{1}}(v\theta-b,v\theta-b+1;-x')
\end{align*}
we see that inside the integral $\Gamma(v\theta-b;x')$ can be replaced by $\Gamma(v\theta-b)$,
and with the choice of  contour $\gamma$ contains all poles. Making use then of (\ref{U}) gives the stated integral form of  $K_{01}(x,x')$.

The hard edge scaled limits follow upon use of the asymptotic formula for the Gamma function \eqref{eq:asymptotic}.
\end{proof}


\subsection{A different approach --- ratios of average characteristic polynomials}

A method introduced in  \cite{forrester2016} involving  averaged ratios of characteristic polynomials can be used to give an
alternative derivation of Proposition \ref{prop:correlation}. This formalism is necessary to relate the $\theta$-deformation Cauchy two-matrix model to
the  $\theta$-deformation  Bures model. In this section we list details as required for this latter purpose.

We define the generalized partition function
\begin{align}\label{eq:cauchypartition}
\begin{aligned}
Z_{k_1|l_1;k_2|l_2}^{(N,a,b,C)}(\lambda_1,\lambda_1;\kappa_2,\lambda_2)&=\frac{1}{(N!)^2}\int_{\mathbb{R}^N\times\mathbb{R}^N}\frac{\prod_{1\leq j<k\leq N}(x_k-x_j)(y_k-y_j)}{\prod_{j,k=1}^N(x_j+y_k)}\prod_{1\leq j<k\leq N}(x_k^\theta-x_j^\theta)(y_k^\theta-y_j^\theta)\\
&\times\prod_{j=1}^N x_j^ae^{-x_j}y_j^be^{-y_j}\frac{\prod_{i=1}^{l_1}(x_j^\theta-\lambda_{1,i})\prod_{i=1}^{l_2}(y_j^\theta-\lambda_{2,i})}{\prod_{i=1}^{k_1}(x_j-\lambda_{1,i})\prod_{i=1}^{k_2}(y_j-\kappa_{2,i})}d[x]d[y].
\end{aligned}
\end{align}
Note the dependence on $\{x_j^\theta, y_j^\theta \}$ in the characteristic polynomials due to the $\theta$-deformation of the matrix model.
  Notice too that $\Zzzzz{}$ is the partition function of the $\theta$-deformation Cauchy two-matrix model, $Z_{0|1;0|0}(x)$ (or $Z_{0|0;0|1}(y)$) is the $\theta$-deformed Cauchy bi-orthogonal polynomials discussed above and $Z_{1|0;0|0}(y)$ (or $Z_{0|0;1|0}(x)$) is the Cauchy transformation of these polynomials respectively. 
  
 According to results of \cite{kieburg2010}, we know this partition function induces the  determinantal structure 
\begin{align*}
Z_{k_1|l_1;k_2|l_2}^{(N,a,b,C)}(\kappa_1,\lambda_1;\kappa_2,\lambda_2)&=\frac{(-1)^{k_1(k_1-1)/2+k_2(k_2-1)/2+l_1l_2}Z_{0|0;0|0}^{(\tilde{N},a,b,C)}}{B_{k_1|l_1}(\kappa_1,\lambda_1)B_{k_2|l_2}(\kappa_2;\lambda_2)}\\
&\times\det\left[\begin{array}{c|c}
\frac{Z_{1|0;1|0}^{(\tilde{N}+1,a,b,C)}(\kappa_{1,i};\kappa_{2,j})}{\Zzzzz{(\tN,a,b,C)}}&\frac{1}{\kappa_{1,i}-\lambda_{1,j}}\frac{\Zoozz{(\tN,a,b,C)}(\kappa_{1,i},\lambda_{1,j})}{\Zzzzz{(\tN,a,b,C)}}\\
\hline\\
\frac{1}{\kappa_{2,j}-\lambda_{2,i}}\frac{\Zzzoo{(\tN,a,b,C)}(\kappa_{2,j},\lambda_{2,i})}{\Zzzzz{(\tN,a,b,C)}}&-\frac{\Zzozo{(\tN-1,a,b,C)}(\lambda_{1,j};\lambda_{2,i})}{\Zzzzz{(\tN,a,b,C)}}
\end{array}
\right]
\end{align*}
where $\tN=N+l_1-k_1=N+l_2-k_2>1$ and $B_{k|l}(\kappa;\lambda)$ is the Cauchy-Vandermonde determinant
\begin{align*}
B_{k|l}(\kappa;\lambda)=\frac{\Delta_k(\kappa)\Delta_l(\lambda)}{\prod_{i=1}^k\prod_{j=1}^l(\kappa_i-\lambda_j)}.
\end{align*}

Only the cases $k_1=l_1=k$ and $k_2=l_2=l$ are independent as other cases follow by taking appropriate
variables $\kappa_1,\kappa_2,\lambda_1,\lambda_2$ to infinity. 
In addition to this restriction, the case $b=a+1$ plays a special role, as then the partition function can be expressed as the
determinant of an anti-symmetric matrix,
\begin{align}\label{eq:pfsquare}
\begin{aligned}
Z_{k|l;k|l}^{(N,a,a+1,C)}&(\kappa,\lambda;\kappa,\lambda)=\frac{\Zzzzz{(\tN,a,a+1,C)}}{B_{k|l}^2(\kappa,\lambda)}\\
&\times\left\{
\begin{array}{lc}
\det\left[\begin{array}{c|c}
\hat{K}_{11}^{(\tN+1)}(\kappa_i,\kappa_j)&-\hat{K}_{01}^{(\tN)}(\kappa_i,\lambda_j)\\\hline
\hat{K}_{01}^{(\tN)}(\kappa_j,\lambda_i)&\hat{K}_{01}^{(\tN-1)}(\lambda_i,\lambda_j)
\end{array}
\right]&\text{for $k+l$ even,}\\
\det\left[\begin{array}{c|c|c}
\hat{K}_{11}^{(\tN+2)}(\kappa_i,\kappa_j)&-\hat{K}_{01}^{(\tN+1)}(\kappa_i,\lambda_j)&\hat{K}_1^{(\tN+1)}(\kappa_i)\\\hline
\hat{K}_{01}^{(\tN+1)}(\kappa_j,\lambda_i)&\hat{K}_{00}^{(\tN)}(\lambda_i,\lambda_j)&\hat{K}_0^{(\tN)}(\lambda_i)\\\hline
-\hat{K}_1^{(\tN+1)}(\kappa_j)&-\hat{K}_0^{(\tN)}(\lambda_j)&0
\end{array}
\right]&\text{for $k+l$ odd}
\end{array}
\right.
\end{aligned}
\end{align}
with kernels
\begin{align*}
&\hat{K}_{11}^{(N)}(\kappa_1,\kappa_2)=\frac{\Zozoz{(N,a,a+1,C)}(\kappa_1;\kappa_2)-\Zozoz{(N,a,a+1,C)}(\kappa_2;\kappa_1)}{2\Zzzzz{(N-1,a,a+1,C)}},\\
&\hat{K}_{01}^{(N)}(\kappa;\lambda)=\frac{1}{\kappa-\lambda}\frac{\Zzzoo{(N,a,a+1,C)}(\kappa,\lambda)+\Zoozz{(N,a,a+1,C)}(\kappa;\lambda)}{2\Zzzzz{(N,a,a+1,C)}},\\
&\hat{K}_{00}^{(N)}(\lambda_1,\lambda_2)=\frac{\Zzozo{(N,a,a+1,C)}(\lambda_2;\lambda_1)-\Zzozo{(N,a,a+1,C)}(\lambda_1;\lambda_2)}{2\Zzzzz{N+1,a,a+1,C}},\\
&\hat{K}_1^{(N)}(\kappa)=\frac{\Zzzoz{(N,a,a+1,C)}(\kappa)+\Zozzz{(N,a,a+1,C)}(\kappa)}{2\Zzzzz{N,a,a+1,C}},\\
&\hat{K}_0^{(N)}(\lambda)=\frac{\Zzzzo{(N,a,a+1,C)}(\lambda)-\Zzozz{(N,a,a+1,C)}(\lambda)}{2\Zzzzz{(N+1,a,a+1,C)}}.
\end{align*}

Finally, introduce the correlation function with self-energy term,
\begin{align*}
&\hat{\rho}_{k,l}(x_1,\cdots,x_k;y_1,\cdots,y_l)\\&:=\frac{1}{\Zzzzz{(N,a,b,C)}}\frac{1}{(N!)^2}\int_{\mathbb{R}^N\times\mathbb{R}^N}\frac{\prod_{1\leq j<k\leq N}(x'_k-x'_j)(y'_k-y'_j)}{\prod_{j,k=1}^N(x'_j+y'_k)}\prod_{1\leq j<k\leq N}({x'}_k^\theta-{x'}_j^\theta)({y'}_k^\theta-{y'}_j^\theta)\\
&\times\prod_{j=1}^N {x'}_j^ae^{-x'_j}{y'}_j^be^{-y'_j}\prod_{j=1}^k\left(\frac{1}{N}\sum_{i=1}^N\delta(x_j-x'_i)\right)\prod_{j=1}^l\left(\frac{1}{N}\sum_{i=1}^N\delta(y_j-y_i')\right)\\
&=\frac{1}{\Zzzzz{(N,a,b,C)}}\lim_{\epsilon\to0}\sum_{L_j,L_i'=\pm}\prod_{j=1}^k(\frac{L_j}{2\pi iN}\frac{\partial}{\partial \tilde{x}_j})\prod_{i=1}^l(\frac{L'_i}{2\pi iN}\frac{\partial}{\partial \tilde{y}_i})Z_{k|k;l|l}^{(N,a,b,C)}(\tilde{x}+iL\epsilon,x;\tilde{y}+iL'\epsilon,y)|_{\tilde{x}=x,\tilde{y}=y},
\end{align*}
where the equality comes from the residue theorem.
This correlation function with self-energy connects with the correlation functions defined in Proposition \ref{prop:correlation} by the formula
\begin{align}\label{LO}
\hat{\rho}_{k,l}^{(N,a,b,C)}(x;y)=\frac{(N!)^2}{(N-k)!(N-l)!N^{k+l}}\rho_{k,l}^{(N,a,b,C)}(x;y)+\text{lower order terms}.
\end{align}
(the meaning of ``lower order terms" is given in \cite{forrester2016}).

\section{$\theta$-deformation of the Bures ensemble}\label{sec:bures}
We now turn our attention to  the $\theta$-deformation of the Bures ensemble (\ref{E2}). Its relationship to the  $\theta$-deformation Cauchy two-matrix model,
as known from  \cite{forrester2016} in the case $\theta = 1$, is essential to our working.
%
\subsection{The partition function of the $\theta$-deformation Bures ensemble}

Let's firstly consider the partition function of the $\theta$-deformation Bures ensemble,
\begin{align}\label{eq:burespartition}
Z_N^{B}(a;\theta)=\frac{1}{N!}\int_{\mathbb{R}_+^N}\prod_{1\leq j,k\leq N}\frac{x_k-x_j}{x_k+x_j}(x_k^\theta-x_j^\theta)\prod_{j=1}^N{x_j^ae^{-x_j}}dx_j.
\end{align}
The working in \cite{forrester2016} for the case $\theta = 1$ can be used to show that this partition function can be written in a Pfaffian form.

\begin{lemma}
We have
\begin{align*}
Z_N^B(a;\theta)=\left\{\begin{array}{cc}
{\rm Pf}(I^B_{j,k})_{j,k=1}^N& \text{for $N$ even},\\
{\rm Pf}\left(\begin{array}{cc}
0& {i}^B_j\\
-{i}^B_k& I^B_{j,k}\end{array}
\right)_{j,k=1}^N&\text{for $N$ odd},
\end{array}
\right.
\end{align*}
where 
\begin{align}\label{eq:momentsb}
I_{j,k}^B=\int_{\mathbb{R}_+^2}\frac{y-x}{y+x}x^{a+\theta(j-1)}y^{a+\theta(k-1)}e^{-(x+y)}dxdy,\quad i_j^B=\int_{\mathbb{R}_+}x^{a+\theta(j-1)}e^{-x}dx.
\end{align}
\end{lemma}

\begin{proof}
The Schur's Pfaffian identity tells us
\begin{equation}\label{Sc}
\Pf(A):=\prod_{1\leq j<k\leq N}\frac{x_k-x_j}{x_k+x_j}=
\left\{\begin{array}{cc}
\Pf(a_{j,k})_{j,k=1}^N& \text{for $N$ even},\\
\Pf\left(\begin{array}{cc}
0& \mathds{1}_N\\
-\mathds{1}_N^T&(a_{j,k})_{j,k=1}^N
\end{array}
\right)&\text{for $N$ odd},
\end{array}
\right.
\end{equation}
where $\mathds{1}_N$ is a $N$ dimensional row vector of elements $1$ and $a_{j,k}=\frac{x_k-x_j}{x_k+x_j}$. Here we also use the de Bruijn's notation
\cite{dB55} 
 $\Pf(A)$ to denote the Pfaffian of an anti-symmetric matrix of even and odd order. The stated formulas now follow by applying de Bruijn's 
 integration formula when the integrand consists of the product of a Pfaffian and a determinant 
 \cite{dB55} (see also the review \cite{Fo18}).
 \end{proof}
 

One can evaluate this partition function by computing this Pfaffian, however, here we compute this partition function through its relation
to  the partition function of $\theta$-deformation Cauchy two-matrix model. 
\begin{proposition}\label{prop:partition}
We have
\begin{align*}
(Z_N^B(a;\theta))^2=2^NZ_N^C(a,a+1;\theta).
\end{align*}
\end{proposition}

\begin{proof}
This follows from the fact that the 
moments of $\theta$-deformation Cauchy two-matrix model and Bures ensemble are of the form
\begin{align*}
I_{j,k}^C(a,a+1;\theta)&=\int_{\mathbb{R}_+^2}\frac{1}{x+y}x^{a+\theta(j-1)}y^{a+1+\theta(k-1)}e^{-(x+y)}dxdy,\\
I_{j,k}^B(a;\theta)&=\int_{\mathbb{R}_+^2}\frac{y-x}{y+x}x^{a+\theta(j-1)}y^{a+\theta(k-1)}e^{-(x+y)}dxdy,
\end{align*}
and so they are connected with each other by the relation
\begin{align}\label{eq:momentsrelation}
2I_{j,k}^C(a,a+1;\theta)=I_{j,k}^B(a;\theta)+i_j^B(a;\theta)i_k^B(a;\theta).
\end{align}
From here we proceed as in the working given in \cite{bertola2009} in the case $\theta = 1$ to deduce the result.
\end{proof}

\begin{corollary}
We have the evaluation
\begin{align*}
Z_N^B(a;\theta)=\frac{{\pi}^{\frac{N}{2}}}{2^{\theta N^2+2aN-(\theta-1)N}}\times \frac{\Gamma(j+1)\Gamma(2\hat{\beta}+j+2)}{\Gamma(\theta(j+\hat{\beta}+1)+1/2)}\times\sqrt{
\frac{\Gamma(2\theta(\hat{\beta}+j+1))\Gamma(2\theta(\hat{\beta}+j+1)+1)}{\Gamma(2(\hat{\beta}+j+1))\Gamma(2(\hat{\beta}+j+1)+1)}}
\end{align*}
with $\hat{\beta}=\frac{a+1}{\theta}-1$.
\end{corollary}

\begin{proof}
We make use of  the partition function of the evaluation of the $\theta$-deformation Cauchy two-matrix model partition function \eqref{eq:partition}
in Proposition  \ref{prop:partition} and simplify using gamma function identities.
\end{proof}

\subsection{Average of characteristic polynomials------$\theta$-deformation partial-skew-orthogonal polynomials}
This section relies on results from the recent work  \cite{chang2018}.

Let's firstly introduce a skew-symmetric inner product
\begin{align*}
\langle x^i,y^j\rangle_{\text{Bures}}=\int_{\mathbb{R}_+^2}\frac{y-x}{y+x}x^{i\theta+a}y^{j\theta+a+1}e^{-(x+y)}dxdy,
\end{align*}
then follow the definition in \cite{chang2018}, we can define a family of monic polynomials $\{\phi_n(z)\}_{n=0}^\infty$ by considering the condition
\begin{align*}
\langle \phi_{2n}(z),z^i\rangle_{\text{Bures}}=\frac{Z_{2n+2}^B}{Z_{2n}^B}\delta_{i,2n+1},\quad \langle \phi_{2n+1},z^j\rangle_\B=-\frac{Z_{2n+1}^B}{Z_{2n}^B}i_j^B,\quad 0\leq j\leq 2n+1.
\end{align*}
Solving this linear system, one knows $\phi_{2n}(z)$ and $\phi_{2n+1}(z)$ also admit determinantal expressions
\begin{align*}
\phi_{2n}(z)&=\frac{1}{\det\left(I_{j,k}^B\right)_{j,k=0}^{2n-1}}\det\left(
\begin{array}{c}
I_{j,k}^B\\z^k
\end{array}
\right)_{\substack{j=0,\cdots,2n\\ k=0,\cdots,2n+1}}\\
\phi_{2n+1}(z)&=\frac{1}{\det\left(I_{j,k}^B|i_j^B\right)_{\substack{j=0,\cdots,2n+1\\k=0,\cdots,2n}}}\det\left(\begin{array}{cc}
I_{j,k}^B&-i_j^B\\
z^k&0
\end{array}\right)_{j,k=0}^{2n+1}.
 \end{align*}
 
 It is a remarkable fact that the determinant expressions admit the Pfaffian form
 \begin{align*}
\phi_{2n}(z)&=\frac{1}{Z_{2n}^B}\Pf\left(
\begin{array}{cc}
I_{j,k}^B&z^k\\
-z^j&0
\end{array}
\right)_{j,k=0}^{2n}\\
\phi_{2n+1}(z)&=\frac{1}{Z_{2n+1}^B}\Pf\left(
\begin{array}{ccc}
0&i_j^B&0\\
-i_k^B&I_{j,k}^B&z^k\\
0&-z^j&0
\end{array}
\right)_{j,k=0}^{2n+1}
\end{align*}
where the moments $I_{j,k}^B$ and $i_j^B$ are the same as defined in \eqref{eq:momentsb}.
This can be established by employing the Jacobi identity for determinants; see  \cite{chang2018} for details. From the Pfaffian form,
consideration of the de Bruijn formula \cite{dB55} shows that
\begin{align*}
\phi_N(z)=\frac{1}{Z_N^B}\int_{\mathbb{R}_+^N}\prod_{1\leq j<k\leq N}\frac{x_k-x_j}{x_k+x_j}(x_k^\theta-x_j^\theta)\prod_{j=1}^N(z-x_j^\theta)x_j^ae^{-x_j}dx_j.
\end{align*}
The fact that for $A$ an even dimensional anti-symmetric matrix ${\rm Pf} \, A = (\det A)^2$ it also follows from the Pfaffian form that
a unified determinant expression of the square of the Bures polynomials can be given,
\begin{align*}
\phi_n^2(z)=\frac{1}{(Z_n^B)^2}\det\left(\begin{array}{ccc}
1&i_j^B&0\\
-i_k^B&I_{j,k}^B&z^k\\
0&-z^j&0\end{array}
\right)_{j,k=0}^n.
\end{align*}

If we consider the $\theta$-deformation Cauchy bi-orthogonal polynomials ($\theta$-CBOPs) with special inner product
\begin{align*}
I_{j,k}^C=\langle x^i,y^j\rangle_\C=\int_{\mathbb{R}_+^2}\frac{1}{x+y}x^{j\theta+a}y^{k\theta+a+1}e^{-(x+y)}dxdy,
\end{align*}
then we can define the monic $\theta$-CBOPs 
\begin{align*}
\tilde{P}_n(x)=\frac{1}{Z_n^C}\det\left(
I_{j,k}^C\,\,\,x^j
\right)_{\substack{j=0,\cdots,n\\k=0,\cdots,n-1}},\quad\tilde{Q}_n(y)=\frac{1}{Z_n^C}\det\left(
\begin{array}{c}
I_{j,k}^C\\
y^k
\end{array}
\right)_{\substack{j=0,\cdots,n-1\\k=0,\cdots,n}},
\end{align*}
obeying the bi-orthogonal relation $$\langle \tilde{P}_n(x),\tilde{Q}_m(y)\rangle_\C=\frac{Z_{n+1}^C}{Z_n^C}\delta_{n,m}.$$
Furthermore, the $\theta$-CBOPs and $\theta$-Bures polynomials can be related to each other.
\begin{proposition}\label{prop:polynomials}
There exist linear relations between $\theta$-CBOPs and $\theta$-Bures polynomials,
\begin{enumerate}
\item $\tilde{P}_n(x)+\tilde{Q}_n(x)=2\phi_n(x)$;
\item $\tilde{Q}_n(x)=\phi_n(x)-\frac{Z_{n+1}^BZ_{n-1}^B}{(Z_n^B)^2}\phi_{n-1}(x)$;
\item $\tilde{P}_n(x)=\phi_n(x)+\frac{Z_{n+1}^BZ_{n-1}^B}{(Z_n^B)^2}\phi_{n-1}(x)$.
\end{enumerate}
\end{proposition}

\begin{proof}
To prove the first equality, we have to use the relation between the moments of $\theta$-CBOPs and $\theta$-Bures polynomials, from which we know
\begin{align*}
Z_n^C\tilde{Q}_n(x)=\left|\begin{array}{cccc}
I_{0,0}^C&\cdots&I_{0,n}^C\\
\vdots&&\vdots\\
I_{n-1,0}^C&\cdots&I_{n-1,n}^C\\
1&\cdots&x^n\end{array}
\right|&=\frac{1}{2^n}\left|
\begin{array}{cccc}
1&0&\cdots&0\\
-i^B_0&2I_{0,0}^C&\cdots&2I_{0,n}^C\\
\vdots&\vdots&&\vdots\\
-i^B_{n-1}&2I_{n-1,0}^C&\cdots&2I_{n-1,n}^C\\
0&1&\cdots&x^n
\end{array}
\right|\\&=\frac{(-1)^n}{2^n}\left|\begin{array}{ccccc}
1&i_0^B&\cdots&i^B_{n-1}&0\\
i^B_0&I_{0,0}^B&\cdots&I_{0,n-1}&1\\
\vdots&\vdots&&\vdots&\vdots\\
i^B_n&I_{n,0}^B&\cdots&I_{n,n-1}^B&x^n
\end{array}
\right|.
\end{align*}
For $n$ even, we know the above equation is equal to
\begin{align*}
\frac{1}{2^n}\left(
\left|\begin{array}{ccccc}
-1&i^B_0&\cdots&i^B_{n-1}&0\\
i^B_0&I_{0,0}^B&\cdots&I_{0,n-1}^B&1\\
\vdots&\vdots&&\vdots&\vdots\\
i^B_n&I_{n,0}^B&\cdots&I_{n,n-1}^B&x^n
\end{array}
\right|+\left|
\begin{array}{ccccc}
2&0&\cdots&0&0\\
i^B_0&I_{0,0}^B&\cdots&I_{0,n-1}^B&1\\
\vdots&\vdots&&\vdots&\vdots\\
i^B_n&I_{0,0}^B&\cdots&I_{n,n-1}^B&x^n
\end{array}
\right|
\right)=-Z_n^C\tilde{P}_n(x)+\frac{(Z_{n}^B)^2}{2^{n-1}}\phi_n(x),
\end{align*}
and for $n$ odd, one knows the above equality is equal to
\begin{align*}
\frac{-1}{2^n}\left(
\left|\begin{array}{ccccc}
1&-i^B_0&\cdots&-i^B_{n-1}&0\\
i^B_0&I_{0,0}^B&\cdots&I_{0,n-1}^B&1\\
\vdots&\vdots&&\vdots&\vdots\\
i^B_n&I_{n,0}^B&\cdots&I_{n,n-1}^B&x^n
\end{array}
\right|+2\left|\begin{array}{ccccc}
0&i^B_0&\cdots&i^B_{n-1}&0\\
i^B_0&I_{0,0}^B&\cdots&I_{0,n-1}^B&1\\
\vdots&\vdots&&\vdots&\vdots\\
i^B_n&I_{n,0}^B&\cdots&I_{n,n-1}^B&x^n\end{array}
\right|
\right)=-Z_n^C\tilde{P}_n(x)+2Z_n^C\phi_n(x),
\end{align*}
which establishes the first equality.

For the second equality, note
\begin{align*}
Z_n^C\tilde{Q}_n(x)&=\left|\begin{array}{cccc}
I_{0,0}^C&\cdots&I_{0,n}^C\\
\vdots&&\vdots\\
I_{n-1,0}^C&\cdots&I_{n-1,n}^C\\
1&\cdots&x^n\end{array}
\right|=\frac{1}{2^n}\left|
\begin{array}{cccc}
1&0&\cdots&0\\
-i^B_0&2I_{0,0}^C&\cdots&2I_{0,n}^C\\
\vdots&\vdots&&\vdots\\
-i^B_{n-1}&2I_{n-1,0}^C&\cdots&2I_{n-1,n}^C\\
0&1&\cdots&x^n
\end{array}
\right|\\&=\frac{1}{2^n}\left(
\left|\begin{array}{cccc}
0&i^B_0&\cdots&i^B_n\\
-i^B_0&I_{0,0}^B&\cdots&I^B_{0,n}\\
\vdots&\vdots&&\vdots\\
-i^B_{n-1}&I_{n-1,0}^B&\cdots&I_{n-1,n}^B\\
0&1&\cdots&x^n
\end{array}
\right|+\left|\begin{array}{cccc}
1&0&\cdots&0\\
-i^B_0&I_{0,0}^B&\cdots&I_{0,n}^B\\
\vdots&\vdots&&\vdots\\
-i^B_{n-1}&I_{n-1,0}^B&\cdots&I_{n-1,n}^B\\
0&1&\cdots&x^n\end{array}
\right|
\right)\\
&:=\frac{1}{2^n}(A_n+B_n).
\end{align*}
For $n$ even, we know $A_n=Z_{n+1}^BZ_{n-1}^B\phi_{n-1}(x)$, $B_n=(Z_n^B)^2\phi_n(x)$ and for $n$ odd, we know $A_n=(Z_n^B)^2\phi_n(x)$ and $B_n=Z_{n+1}^BZ_{n-1}^B\phi_{n-1}(x)$. The second equality follows. Note the third equality is a linear combination of the first two equations; the details of
its derivation are omitted.
\end{proof}

\begin{remark}
The linear relation between the $\theta$-CBOPs and $\theta$-Bures polynomials
revealed in this proposition implies that if one computes the Pfaffian point process in terms of $\theta$-Bures polynomials, then this kernel can be expressed as $\theta$-CBOPs as well. This is the key to connecting these two different ensembles.
\end{remark}

\subsection{Connections between $\theta$-deformation Cauchy two-matrix model and $\theta$-deformation Bures ensemble}
We start with the generalised partition function of the $\theta$-deformation Bures ensemble
\begin{align}\label{eq:burespartition2}
Z_{k|l}^{(N,a,B)}(\kappa,\lambda)=\frac{1}{N!}\int_{\mathbb{R}_+^N}\prod_{1\leq j,k\leq N}\frac{x_k-x_j}{x_k+x_j}(x_k^\theta-x_j^\theta)\prod_{j=1}^N\frac{\prod_{i=1}^l(x_j^\theta-\lambda_i)}{\prod_{i=1}^k(x_j-\kappa_i)}{x_j^ae^{-x_j}}dx_j
\end{align}
(cf.~(\ref{eq:cauchypartition})).
Results from \cite{kieburg20102} tell us that this can be written in the Pfaffian form
\begin{align}\label{eq:bures1}
\begin{aligned}
Z_{k|l}^{(N,a,B)}(\kappa,\lambda)&=(-1)^{k(k-1)/2+l(l-1)/2}\frac{Z_{0|0}^{(\tN,a,B)}}{B_{k|l}(\kappa,\lambda)}\\&\times\Pf\left[
\begin{array}{c|c}
(\kappa_i-\kappa_j)\frac{Z_{2|0}^{(\tN+2,a,B)}(\kappa_i,\kappa_j)}{Z_{0|0}^{(\tN,a,B)}}&\frac{1}{\kappa_i-\lambda_j}\frac{Z_{1|1}^{(\tN,a,B)}(\kappa_i,\lambda_j)}{Z_{0|0}^{(\tN,a,B)}}\\\hline
\frac{1}{\lambda_i-\kappa_j}\frac{Z_{1|1}^{\tN,a,B}(\kappa_i,\lambda_j)}{Z_{0|0}^{(\tN,a,B)}}&(\lambda_i-\lambda_j)\frac{Z_{0|2}^{(\tN-2,a,B)}(\lambda_i,\lambda_j)}{Z_{0|0}^{(\tN,a,B)}}
\end{array}\right],\quad \text{for $k+l$ even}
\end{aligned}
\end{align}
\begin{align}\label{eq:bures2}
\begin{aligned}
&Z_{k|l}^{(N,a,B)}(\kappa,\lambda)=(-1)^{k(k-1)/2+l(l-1)/2}\frac{Z_{0|0}^{(\tN+1,a,B)}}{B_{k|l}(\kappa,\lambda)}\\&\times\Pf\left[
\begin{array}{c|c|c}
(\kappa_i-\kappa_j)\frac{Z_{2|0}^{(\tN+3,a,B)}(\kappa_i,\kappa_j)}{Z_{0|0}^{\tN+1,a,B}}&\frac{1}{\kappa_i-\lambda_j}\frac{Z_{1|1}^{\tN+1,a,B}(\kappa_i,\lambda_j)}{Z_{0|0}^{(\tN+1,a,B)}}&\frac{Z_{1|0}^{\tN+1,a,B}(\kappa_i)}{Z_{0|0}^{(\tN+1,a,B)}}\\\hline
\frac{1}{\lambda_i-\kappa_j}\frac{Z_{1|1}^{\tN+1,a,B}(\kappa_i,\lambda_j)}{Z_{0|0}^{(\tN+1,a,B)}}&(\lambda_i-\lambda_j)\frac{Z_{0|2}^{(\tN-1,a,B)}(\lambda_i,\lambda_j)}{Z_{0|0}^{(\tN+1,a,B)}}&\frac{Z_{0|1}^{(\tN-1,a,B)}(\lambda_i)}{Z_{0|0}^{\tN+1,a,B}}\\\hline
-\frac{Z_{1|0}^{\tN+1,a,B}(\kappa_i)}{Z_{0|0}^{(\tN+1,a,B)}}&-\frac{Z_{0|1}^{(\tN-1,a,B)}(\lambda_i)}{Z_{0|0}^{\tN+1,a,B}}&0
\end{array}
\right],\quad\text{for $k+l$ odd}.
\end{aligned}
\end{align}

Moreover, by generalising the moments relation between the $\theta$-deformation Cauchy ensemble and Bures ensemble \eqref{eq:momentsrelation}
according to the working of \cite[Appendix A]{forrester2016}, the partition function identity of Proposition \ref{prop:partition} can be similarly
generalised.

\begin{proposition}\label{P3.6}
We have
\begin{align*}
(Z_{k|l}^{(N,a,B)}(\kappa,\lambda))^2=2^NZ_{k|l;k|l}^{(N,a,a+1,B)}(\kappa,\lambda;\kappa,\lambda).
\end{align*}
\end{proposition}
As a consequence, we can compare the expressions \eqref{eq:bures1} and \eqref{eq:bures2} with \eqref{eq:pfsquare}, and obtain  linear relationship between these average of characteristic polynomials.

\begin{corollary}
Let $\hat{Z}^{(N,a,a+1,C)}$ be the mean value of the difference of the two variable sets of the $\theta$-deformed Cauchy-Laguerre ensemble
\begin{align*}
\hat{Z}^{(N,a,a+1,C)}&=\frac{1}{(N!)^2}\int_{\mathbb{R}^N\times\mathbb{R}^N}\frac{\prod_{1\leq j<k\leq N}(x_k-x_j)(y_k-y_j)}{\prod_{j,k=1}^N(x_j+y_k)}\prod_{1\leq j<k\leq N}(x_k^\theta-x_j^\theta)(y_k^\theta-y_j^\theta)\\&\times
\prod_{j=1}^N x_j^ae^{-x_j}y_j^{a+1}e^{-y_j}\sum_{j=1}^N(x_j-y_j)dx_jdy_j,
\end{align*}
We have
\begin{align*}
&\frac{Z_{2|0}^{(N,a,B)}(\kappa_1,\kappa_2)}{Z_{0|0}^{N,a,B}}=\frac{1}{\kappa_2-\kappa_1}\frac{\Zozoz{(N-1,a,a+1,B)}(\kappa_1;\kappa_2)-\Zozoz{(N-1,a,a+1,B)}(\kappa_2;\kappa_1)}{\hat{Z}^{(N-1,a,a+1,C)}},\\
&\frac{Z_{0|2}^{(N,a,B)}(\lambda_1,\lambda_2)}{Z_{0|0}^{(N,a,B)}}=\frac{1}{\lambda_1-\lambda_2}\frac{\Zzozo{(N+1,a,a+1,C)}(\lambda_1;\lambda_2)-\Zzozo{(N+1,a,a+1,C)}(\lambda_2;\lambda_1)}{\hat{Z}^{(N+1,a,a+1,C)}},\\
&\frac{Z_{1|1}^{(N,a,B)}(\kappa,\lambda)}{Z_{0|0}^{(N,a,B)}}=\frac{\Zzzoo{(N,a,a+1,C)}(\kappa,\lambda)+\Zoozz{(N,a,a+1,C)}(\kappa,\lambda)}{2\Zzzzz{(N,a,a+1,C)}},\\
&\frac{Z_{1|0}^{N,a,B}(\kappa)}{Z_{0|0}^{N,a,B}}=\frac{\Zzzoz{(N,a,a+1,C)}(\kappa)+\Zozzz{(N,a,a+1,C)}(\kappa)}{2\Zzzzz{(N,a,a+1,C)}},\\
&\frac{Z_{0|1}^{(N,a,B)}(\lambda)}{Z_{0|0}^{N,a,B}}=\frac{\Zzozz{(N+1,a,a+1,C)}(\lambda)-\Zzzzo{(N+1,a,a+1,C)}(\lambda)}{\hat{Z}^{(N+1,a,a+1,C)}}.
\end{align*}
\end{corollary}

\begin{remark}
The last two equations correspond to the results of Proposition \ref{prop:polynomials}.
\end{remark}

\subsection{Correlation function of $\theta$-deformation Bures ensemble}
Now we can compute the correlation function of $\theta$-deformation Bures ensemble in terms of the average of characteristic polynomials.

Let us firstly consider the $k$-point correlation function 
\begin{align*}
\rho^{(N,a,B)}_k(x_1,\cdots,x_k)=\frac{1}{(N-k)!Z_N^B}\int_{\mathbb{R}_+^{N-k}}\prod_{1\leq j<k\leq N}\frac{x_k-x_j}{x_k+x_j}(x_k^\theta-x_j^\theta)\prod_{j=k+1}^N x_j^ae^{-x_j}dx_j.
\end{align*}
If we specify the correlation function including ``self-energy'' term by the definition
\begin{align*}
\hat{\rho}_k^{(N,a,B)}(x):&=\frac{1}{Z_{0|0}^{(N,a,B)}}\frac{1}{N!}\int_{\mathbb{R}_+^N}\frac{x'_k-x'_j}{x'_k+x'_j}({x'}_k^\theta-{x'}_j^\theta)\prod_{j=1}^N{x_j^ae^{-x_j}}dx_j\prod_{j=1}^k\left(\frac{1}{N}\sum_{i=1}^N\delta(x_j-x'_i)
\right)\\
&=\frac{1}{Z_{0|0}^{(N,a,B)}}\lim_{\epsilon\to0}\sum_{L_j=\pm}\prod_{j=1}^k(\frac{L_j}{2\pi i N}\frac{\partial}{\partial \tilde{x}_j})Z_{k|k}^{(N,a,B)}(\tilde{x}+iL\epsilon,x)|_{\tilde{x}=x},
\end{align*}
then these two correlation functions relate to each other by
\begin{align*}
\hat{\rho}_k^{(N,a,B)}(x)=\frac{N!}{(N-k)!N^k}\rho_k^{(N,a,B)}+\text{lower order terms}
\end{align*}
(cf.~(\ref{LO})).
Moreover, from the relationship of Proposition \ref{P3.6}
\begin{align*}
(Z_{k|l}^{(N,a,B)}(\kappa,\lambda))^2=2^NZ_{k|l;k|l}^{(N,a,a+1,B)}(\kappa,\lambda;\kappa,\lambda),
\end{align*}
we know 
\begin{align*}
\hat{\rho}_1^{(N,a,B)}(z)&=\frac{1}{2}(\rho_{1,0}^{(N,a,a+1,C)}(z)+\rho_{0,1}^{(N,a,a+1,C)}(z)),\\
\hat{\rho}_2^{(N,a,B)}(z)&=\frac{1}{2}\big[\hat{\rho}_{2,0}^{(N,a,a+1,C)}(z_1,z_2)+\hat{\rho}_{1,1}^{(N,a,a+1,C)}(z_1,z_2)+\hat{\rho}_{1,1}^{(N,a,a+1,C)}(z_2,z_1)+\hat{\rho}_{0,2}^{(N,a,a+1,C)}(z_1,z_2)\\&-\frac{1}{2}(\rho_{1,0}^{N,a,a+1,C}(z_1)+\rho_{0,1}^{(N,a,a+1,C)}(z_1))(\rho_{1,0}^{N,a,a+1,C}(z_2)+\rho_{0,1}^{(N,a,a+1,C)}(z_2))\big],
\end{align*}
from which we can deduce following proposition.

\begin{proposition}
Let $z_1,\cdots,z_k\in\mathbb{R}_+$ be pairwise different. 
Introduce the notation 
\begin{align*}
&\Delta K_{11}^{(N,a,a+1,C)}(z_i;z_j)=\hat{K}_{11}^{(N,a,a+1,C)}(z_i;z_j)-\hat{K}_{11}^{(N,a,a+1,C)}(z_j;z_i),\\
&\sum K_{01}^{(N,a,a+1,C)}(z_i;z_j)=\hat{K}_{01}^{(N,a,a+1,C)}(z_i;z_j)+\hat{K}_{10}^{(N,a,a+1,C)}(z_i;z_j),\\
&\Delta K_{00}^{(N,a,a+1,C)}(z_i;z_j)=\hat{K}_{00}^{(N,a,a+1,C)}(z_i;z_j)-\hat{K}_{00}^{(N,a,a+1,C)}(z_j;z_i).
\end{align*}
Then the above correlation function can be expressed as a Pfaffian
\begin{align*}
\rho_k^{(N,a,B)}(z_1,\cdots,z_k)=&(-1)^{k(k-1)/2}\frac{N!}{(2N)^k(N-k)!}\\
&\times\Pf\left[\begin{array}{cc}
\Delta K_{11}^{(N,a,a+1,C)}(z_i;z_j)&\sum K_{01}^{(N,a,a+1,C)}(z_i;z_j)\\
-\sum K_{01}^{(N,a,a+1,C)}(z_i;z_j)&\Delta K_{00}^{(N,a,a+1,C)}(z_j;z_i)
\end{array}
\right]_{1\leq i,j\leq k}
\end{align*}
\end{proposition}

Hence the kernels of the correlation functions for $\theta$-deformation Bures ensemble are exactly expressed as the kernels of that for $\theta$-deformation Cauchy two-matrix model. Moreover, from the asymptotic behaviors of $\theta$-deformation Cauchy two-matrix model, one can write down the hard edge scaling limit of the $k$-point correlation function of the Bures ensemble.

\begin{proposition}
Let $z_1,\cdots,z_k\in\mathbb{R}_+$ be pairwise different. Then the hard edge scaling limit of the $k$-point correlation function of Bures ensemble is
\begin{align*}
\rho_k^{(\infty,a,B)}(z_1,\cdots,z_k)&=\lim_{N\to\infty}N^{-\frac{2k}{\theta}}\rho_k^{(N,a,B)}(\frac{z_1}{N^{\frac{2}{\theta}}},\cdots,\frac{z_k}{N^{\frac{2}{\theta}}})\\
&=\frac{(-1)^{k(k-1)/2}}{2^k}\Pf\left[
\begin{array}{cc}
\Delta K_{11}^{(\infty,a)}(z_i;z_j)&\sum K_{01}^{(\infty,a)}(z_i;z_j)\\
-\sum K_{01}^{(\infty,a)}(z_i;z_j)&\Delta K_{00}^{(\infty,a)}(z_j;z_i)
\end{array}
\right]_{1\leq i,j\leq k},
\end{align*}
where the kernels are
\begin{align*}
&\Delta K_{00}^{(\infty,a)}(z_i;z_j)=\lim_{N\to\infty}N^{-\frac{4(a+1)}{\theta}}\Delta K_{00}^{(N,a)}(\frac{z_i}{N^{\frac{2}{\theta}}},\frac{z_j}{N^{\frac{2}{\theta}}})\\
&\qquad=\theta\left(\int_0^1t^{\frac{2(a+1)}{\theta}-1}G^\theta_{\infty,a}(tz_i^\theta)G^\theta_{\infty,a}(tz_j^\theta)dt-|_{z_i\leftrightarrow z_j} \right),\\
&\Delta K_{11}^{(\infty,a)}(z_i;z_j)=\lim_{N\to\infty}N^{-\frac{4a}{\theta}}\Delta K_{11}^{(N,a)}(\frac{z_i}{N^{\frac{2}{\theta}}},\frac{z_j}{N^{\frac{2}{\theta}}})\\
&\qquad=\theta (z_iz_j)^{2a+1}\left(
\int_0^1t^{\frac{2(a+1)}{\theta}-1}\tilde{G}^\theta_{\infty,a}(tz_j^\theta)\tilde{G}^\theta_{\infty,a}(tz_i^\theta)dt-|_{z_i\leftrightarrow z_j} 
\right),\\
&\sum K_{01}^{(\infty,a)}(z_i;z_j)=\lim_{N\to\infty}N^{-\frac{2}{\theta}}\sum K_{01}^{(N,a)}(\frac{z_i}{N^{\frac{2}{\theta}}},\frac{z_j}{N^{\frac{2}{\theta}}})\\
&\qquad=\theta\left(
z_i^{2a+1}\int_0^1t^{\frac{2(a+1)}{\theta}-1}G^\theta_{\infty,a}(tz_i^\theta)\tilde{G}^\theta_{\infty,a+1}(ty^\theta)dt
+z_j^{2a+1}\int_0^1 t^{\frac{2(a+1)}{\theta}-1}\tilde{G}^\theta_{\infty,a}(tz_i^\theta)G^\theta_{\infty,a+1}(tz_j^\theta)dt
\right)
\end{align*}
with $G^\theta_{\infty,a}(x)$ and $\tilde{G}^\theta_{\infty,a}(x)$ defined in \eqref{eq:kernel}.
\end{proposition}
\begin{remark}
The kernels $G^\theta_{\infty,a}(x)$ and $\tilde{G}^\theta_{\infty,a}(x)$ are in the Fox H-function class according to the formulas \eqref{tginf} and \eqref{ginf}.
\end{remark}
\section{Conclusion and discussion}
The Bures ensemble was originally introduced into random matrix theory as part of the theory of the Bures metric, which is a natural choice in
measuring the statistical distance between density operators defining quantum states. It has been known since the work of
Sommers and Zyczkowski \cite{SZ04} that the global density, scaled to have support on $[0,b]$ with $b=3\sqrt{3}/2$, is given by
\begin{equation}\label{sz1}
\rho(x) = {1 \over 2 \pi \sqrt{3}}
\bigg ( \Big ( {b \over x} + \sqrt{ {b^2 \over x^2} - 1} \Big )^{2/3} -
 \Big ( {b \over x} - \sqrt{ {b^2 \over x^2} - 1} \Big )^{2/3} \bigg ).
\end{equation}
 The leading asymptotic form near $x=0$ is ${\sqrt{3} \over 2 \pi} {1 \over x^{2/3}}$. 
 
 Of particular relevance to the $\theta$-deformation of the present work is the fact that the moments of the density form a Raney
 sequence (\ref{Ra}) with $(p,r) = (3/2,1/2)$.  For general parameters $(p,r)$ denote the corresponding scaled density by $\rho^{(p,r)}(x)$, which is 
 known to be supported
 on $(0,K_p)$ with $K_p = p^p ( p - 1)^{1-p}$. Moreover, introducing the Stieltjes transform
 $$
 G^{(p,r)}(z) := \int_0^{K_p} {1 \over z - x} \rho^{(p,r)}(x) \, dx =
 {1 \over z} \sum_{n=0}^\infty {1 \over z^n} R_{p,r}(n),
 $$
 it is known from \cite{MPZ13} that $w(z) := z G_{p,r}(z)$ satisfies the algebraic equation
 $$
 w^{p \over r} - z w^{1 \over r} + z = 0;
 $$
 note that in the case $(p,r) = (3/2,1/2)$ this is a cubic, in keeping with the structure of (\ref{sz1}). The cases $(p,r) = (\theta + 1, 1)$ are known to be
 realised by the Laguerre Muttalib--Borodin model \cite{FW15}. Results from \cite{FL15,FLZ15} tell us that the sequence 
 $(p,r) = (\theta/2 + 1, 1/2)$ is realised by the $\theta$-deformed Bures ensemble (\ref{E2}). Generally the Raney density (assuming $r < p$)
 has the small $x$ form \cite{FL15}
 $$
 \rho^{(p,r)}(x)  \sim {1 \over \pi} \Big ( \sin {r \pi \over p} \Big ) x^{- (p - r)/p}.
 $$
 Thus the critical state near the origin is distinct when comparing the  Laguerre Muttalib--Borodin model and
 $\theta$-deformed Bures ensemble.
 
 In keeping with this latter fact, in the present work we have studied the hard edge scaled limit of the  $\theta$-deformed Bures ensemble.
 For this purpose, we first considered a  $\theta$-deformed Cauchy two-matrix model.
 In contrast to the hard edge scaled limit of the  Laguerre Muttalib--Borodin ensemble, which gave rise to a determinantal  point process with a kernel
 given in terms of Wright Bessel functions (or Meijer $G$ functions for $\theta,  1/\theta \in \mathbb Z_+$) the
 hard edge scaled limit of the  $\theta$-deformed Bures ensemble gives rise to a Pfaffian point process with kernel given in terms of
 Fox H-functions. Integrability aspects of this kernel, for example the interpretation of the corresponding hard edge gap probability as a tau function,
 the formulation of its asymptotics in terms of a Riemann-Hilbert problem (see \cite{CGS17} in the case of  Muttalib--Borodin ensemble),
 or associated nonlinear equations (see \cite{St15,WF17,MF18} for previous work) 
 remain for later work.

\section*{Acknowledgements} 
The authors would like to thank Dr. Dang-Zheng Liu and Lun Zhang for helpful comments. This work is part of a research program supported by the Australian Research Council (ARC) through the ARC Centre of Excellence for Mathematical and Statistical frontiers (ACEMS). PJF also acknowledges partial support from ARC grant DP170102028. 

\small
\bibliographystyle{abbrv}

\begin{thebibliography}{1}

\bibitem{AIK13}
G.~Akemann, J.~R. Ipsen, and M.~Kieburg, \emph{Products of rectangular random
  matrices: {S}ingular values and progressive scattering}, Phys. Rev. E
  \textbf{88} (2013), 052118 (13 pages).

\bibitem{beals2013}
R. Beals and J. Szmigielski, \emph{Meijer G-functions: a gentle introduction},
Notices of the AMS \textbf{60} (2013), 866--872.

\bibitem{bertola2009}
M.~Bertola, M.~Gekhtman, and J.~Szmigielski,
\emph{The {C}auchy two-matrix model}, 
Commun. Math. Phys. \textbf{287} (2009), 983--1014.

\bibitem{bertola2010}
M.~Bertola, M.~Gekhtman, and J.~Szmigielski,
\emph{Cauchy biorthogonal polynomials}, 
J. Approx. Theory  \textbf{162} (2010), 832--867.

\bibitem{bertola2014}
M.~Bertola, M.~Gekhtman, and J.~Szmigielski,
\emph{Cauchy-{L}aguerre two-matrix model and the {M}eijer-{G} random point
  field},  Commun. Math. Phys.  \textbf{326} (2014), 111--144.

\bibitem{Bor99}
A.~Borodin, \emph{Biorthogonal ensembles}, Nucl. Phys. B \textbf{536} (1998),
704--732.

 \bibitem{dB55}
N.G. de Bruijn,
\emph{On some multiple integrals involving determinants},
J. Indian Math. Soc., \textbf{19} (1955), 133--151.


\bibitem{chang2018}
X. Chang, Y. He, X. Hu and S. Li,
\emph{Partial-Skew-Orthogonal polynomials and related integrable lattices with Pfaffian tau-functions},
Commun. Math. Phys., DOI: 10.1007/s00220-018-3273-y, 2018.


\bibitem{CR14}
T. Claeys, and S. Romano,
\emph{Biorthogonal ensembles with two-particle interactions},
Nonlinearity 27 (2014) 2419.

\bibitem{CGS17}
T.~Claeys, M.~Girotti, D.~Stivigny,
\emph{Large gap asymptotics at the hard edge for product random matrices and Muttalib-Borodin ensembles}, IMRN,
rnx202 (2017), https://doi.org/10.1093/imrn/rnx202

\bibitem{Fo10}
P.~J. Forrester, \emph{Log-gases and random matrices}, Princeton University
  Press, Princeton, NJ, 2010.

\bibitem{Fo14}
P. J. Forrester,
\emph{Eigenvalue statistics for for product complex Wishart matrices.}
J. Phys. A 47 (2014) 345202.

\bibitem{Fo18}
P. J. Forrester,
\emph{Meet Andr\'eief, Bordeaux 1886, and Andreev, Kharkov 1882--1883}
arXiv:1806.10411.

\bibitem{forrester2016}
P. J.~ Forrester and M. Kieburg,
\emph{Relating the Bures measure to the Cauchy two-matrix model},
Commun. Math. Phys. \textbf{342} (2016), 151--187.

\bibitem{FL15}
P. J. Forrester, and D.-Z. Liu,
\emph{Raney distributions and random matrix theory.}
J. Stat. Phys. 158 (2015) 1051.


\bibitem{FL16}
P. J. Forrester, and D.-Z. Liu,
\emph{Singular values for products of complex Ginibre matrices with a source: hard edge limit and phase transition}. Commun. Math. Phys. \textbf{344} (2016), 333.



\bibitem{FLZ15}
P. J. Forrester, D.-Z. Liu, and P. Zinn-Justin,
\emph{Equilibrium problems for Raney densities.}
Nonlinearity \textbf{28} (2015) 2265.

\bibitem{FW15}
P.~J. Forrester and D.~Wang, \emph{Muttalib--Borodin ensembles in random matrix
  theory --- realisations and correlation functions}, Elec. J. Probab.
  \textbf{22} (2017), 54 (43 pages).
  
\bibitem{fox61}
C.~Fox,
\emph{The G and H functions as symmetrical Fourier kernels}, Trans. Amer. Math. Soc. \textbf{98} (1961), 395-429.
  
  \bibitem{IOTZ06}
  M.~Ishikawa, S.~Okanda, H.~Tagawa, and J.~Zeng,
\emph{Generalizations of Cauchy's determinant
and Schur's Pfaffians},
Adv. in Appl. Math. \textbf{36} (2006), 251-287.

\bibitem{kieburg2010}
M. Kieburg and T. Guhr,
\emph{Derivation of determinantal structures for random matrix ensembles in a new way},
J. Phys. A \textbf{43} (2010), 075201.


\bibitem{kieburg20102}
M. Kieburg and T. Guhr,
\emph{A new approach to derive Pfaffian structures for random matrix ensembles}, 
J. Phys. A, \textbf{43} (2010), 135204.

\bibitem{KKS15}
M. Kieburg, A. B. J. Kuijlaars, and D. Stivigny,
\emph{Singular value statistics of matrix products with truncated unitary matrices.}
Int. Math. Res. Not. (2016) 3392.

\bibitem{Ko89}
I. K.~Kostov, 
\emph{$O(n)$ vector model on a planar random lattice: spectrum of anomalous dimensions},
Mod. Phys. Lett. \textbf{4} (1989), 217--226.

\bibitem{KZ14}
A.~B.~J. Kuijlaars and L.~Zhang, \emph{Singular values of products of {G}inibre
  matrices, multiple orthogonal polynomials and hard edge scaling limits},
  Comm. Math. Phys. \textbf{332} (2014), 759--781.


\bibitem{li2018}
C.~Li and S.~Li,
\emph{The Cauchy two-matrix model, C-Toda lattice and CKP hierarchy},
J. Nonlinear Sci., DOI: 10.1007/s00332-018-9474-x, 2018.


\bibitem{lundmark2005}
H.~Lundmark and J.~Szmigielski,
\emph{Degasperis-{P}rocesi peakons and the discrete cubic string},
Int. Math. Res. Pap. \textbf{2005} (2005), 53--116.

\bibitem{MF18}
V. V.~Mangazeev and P. J. Forrester, \emph{Integrable structure of products of finite complex Ginibre random matrices}, {Physica D}, available online (2018)

\bibitem{MPZ13}
W. M{\l}tokowski,  K.A. Penson,  and K. Zyczkowski, 
\emph{Densities of the  Raney distributions},
Documenta Math. \textbf{18} (2013), 1573-1596.

\bibitem{muttalib1995}
K.~Muttalib,
\emph{Random matrix models with additional interactions},
J. Phys. A \textbf{28} (1995) L159--L164.

\bibitem{PS11}
L.~Pastur and M.~Shcherbina, \emph{Eigenvalue distribution of large random
  matrices}, American Mathematical Society, Providence, RI, 2011.


\bibitem{rains2000}
E. Rains,
\emph{Correlation functions for symmetrized increasing subsequences}, 
arXiv:math/0006097v1 [math.CO]

\bibitem{SZ03}
H.-J. Sommers and K.~\.Zyczkowski, \emph{Bures volume of the set of mixed quantum
  states}, J.~Phys. A \textbf{36}  (2003), 10083--10100.
  
  \bibitem{SZ04} 
  H.-J. Sommers and K.~\.Zyczkowski, \emph{Statistical properties of random density matrices},
  J.~Phys.~A {\bf 37} (2004), 8457--8466.

    \bibitem{St15}
 E.~Strahov, \emph{Dynamical correlation functions for products of random matrices},
Random Matrices Theory and Appl.
  \textbf{4} (2015),    1550020.
  
  \bibitem{WF17} N.S.~Witte and P.J.~Forrester, Singular values of products of random matrices,
Studies Appl.~Math. {\bf 138}, 135--184 (2017) 


\end{thebibliography}

\end{document}